\documentclass{article}
\usepackage{bbm}
\usepackage{float}
\usepackage{amsmath}
\usepackage{graphicx}
\usepackage{natbib}
\usepackage{amssymb}
\usepackage{amsthm}
\usepackage{mathrsfs}
\usepackage{subfigure}
\usepackage{amsfonts} 
\usepackage{docmute}
\usepackage{booktabs}
\usepackage{listings}
\usepackage[autostyle]{csquotes}
\usepackage{hyperref}
\usepackage{url} 
\usepackage{xcolor}
\usepackage{comment}
\usepackage{algorithm,algpseudocode}
\usepackage{spverbatim}
\newcommand{\<}{\langle}
\renewcommand{\>}{\rangle}
\graphicspath{ {./images/} }

\providecommand{\norm}[1]{\lVert#1\rVert}
\providecommand{\abs}[1]{|#1|}

\usepackage[inline]{enumitem}

\theoremstyle{thmstyleone}%
\newtheorem{theorem}{Theorem}
\newtheorem{proposition}[theorem]{Proposition}%
\theoremstyle{thmstyletwo}%
\newtheorem{example}{Example}%
\newtheorem{remark}{Remark}%
\theoremstyle{thmstylethree}%
\newtheorem{definition}{Definition}

\newtheorem{assumption}{Assumption}
\newtheorem{lemma}{Lemma}
\newtheorem{recipe}{Recipe}
\providecommand{\norm}[1]{\lVert#1\rVert}

\newcommand{\blind}{0}
\allowdisplaybreaks
\date{}
\begin{document}

\if0\blind
{
  \title{\bf Wasserstein complexity penalization priors: a new class of penalizing complexity priors}
  \author{David Bolin, Alexandre B. Simas, and Zhen Xiong \hspace{.2cm}\\
    Computer, Electrical and Mathematical Sciences and Engineering \\ Division, 
    King Abdullah University of Science and Technology\\
    Thuwal 23955-6900, Saudi Arabia}
  \maketitle
} \fi

\if1\blind
{
  \bigskip
  \bigskip
  \bigskip
  \begin{center}
    {\LARGE\bf Title}
\end{center}
  \medskip
} \fi

\bigskip
\begin{abstract}
	Penalizing complexity (PC) priors provide a principled framework for reducing model complexity by penalizing the Kullback--Leibler Divergence (KLD) between a ``simple'' base model and a more complex model. However, constructing priors by penalizing the KLD becomes impossible in many cases because the KLD is infinite, and alternative principles often lose interpretability in terms of KLD.
    We propose a new class of priors, the Wasserstein complexity penalization (WCP) priors, which replace the KLD with the Wasserstein distance in the PC prior framework. WCP priors avoid the issue of infinite model distances and retain interpretability by adhering to adjusted principles. Additionally, we introduce the concept of base measures, removing the parameter dependency on the base model, and extend the framework to joint WCP priors for multiple parameters. These priors can be constructed analytically and we have both analytical and numerical implementations in \texttt{R}. We demonstrate their use in previous PC prior applications and as well as new multivariate settings.
\end{abstract}


\section{Introduction}
\label{sec:intro}
Priors are an integral part of the Bayesian inference procedure. 
When direct prior information is available, prior distributions are used to approximate and summarize that information.
This class of priors is known as subjective priors, see, e.g., \citet[Sections 3.2 and 4.2]{robert2007bayesian}, \citet[Section 3.2]{berger1985statistical}.
When 
one does not want the results to be influenced by prior information, 
noninformative priors are used.
There is no unified definition of these priors 
\citep[Section~3.3.1]{berger1985statistical};
however, the general idea is to give no preference to any specific part of the parameter space.
\citet{marquis1820theorie} proposed to assign a uniform distribution on the parameter space with a principle called the principle of insufficient reason, which is regarded as the first noninformative prior.
One of the criticisms of this prior \citep[Section~3.5.1]{robert2007bayesian} is the lack of parameterization invariance.
By letting the prior be proportional to the square root of Fisher information of the parameter, Jeffrey's prior \citep{jeffreys1946invariant} has the property of parameterization invariance while maintaining the idea of uniformity 
\citep[][Section~2.3.1]{kass1989geometry}.
Later, reference priors have been developed as an extension of Jeffrey's prior by formalizing the notion of an objective prior \citep{bernardo1979reference, berger1989estimating, berger1992ordered, berger1992development}. 
See \citet{consonni2018prior} for a survey of these topics.

Weakly informative priors are meant to lie between subjective and  noninformative priors. An example is the maximum entropy method \citep{jaynes1968prior,Jaynes1983PapersOP}, which is useful when partial prior information exists, such as central moments or quantiles of a prior distribution.
Another example that has received much attention recently are Penalizing Complexity (PC) priors \citep{PCpriororigin}, which we review next.
Let $M = \{\mu_\theta : \theta \in \Theta\}$ denote a set of statistical models (probability measures) indexed by a parameter $\theta \in \Theta \subset \mathbb{R}^d$, where $d \geq 1$. 
For instance, $\mu_\theta$ could represent a univariate distribution, such as a generalized Pareto distribution with tail index $\theta \in \mathbb{R}$. Alternatively, $\mu_\theta$ might represent the multivariate Gaussian distribution of an autoregressive process $\{X_t\}_{t=1}^n$. Further, let $\overline{\Theta}$ denote the set of parameters in $\mathbb{R}$ corresponding to probability distributions that can be approximated by models in $M$ according to some criteria chosen when developing the PC prior. The construction of the prior for $\theta$ is guided by the following principles:

\noindent\textbf{1. Preference for Simpler Models:} The prior should favor simpler models. Simplicity is defined in relation to a base model, \(\mu_{\theta_b}\), which is considered the simplest within the extended class of models \(\overline{M} = \{\mu_{\theta} : \theta \in \overline{\Theta}\}\). The prior assigned to \(\theta\) should decrease as the complexity of the model increases, effectively penalizing deviations from the base model.
    
\noindent\textbf{2. Complexity Measurement via KLD:} The Kullback-Leibler Divergence (KLD) \citep{kullback1951information} is used as a measure of complexity and \(d(\theta) = \text{KLD}(\mu_{\theta} \parallel \mu_{\theta_b})\) is used as a ``distance'' between a flexible model \(\mu_{\theta}\)  and  \(\mu_{\theta_b}\).
    
\noindent\textbf{3. Constant Penalization Rate:} Deviations are penalized at a constant rate, with the prior density $\pi_{d(\theta)}$ of the distance $d(\theta)$ satisfying $\pi_{d(\theta)}(d+\delta) = r^\delta \pi_{d(\theta)}(d)$ for $d, \delta > 0$ and a  decay rate $r \in (0,1)$. This leads to an exponential prior $\pi_{d(\theta)}(d) = \eta \exp(-\eta d)$, where $\eta = -\log(r)$, and a change of variables gives the prior for $\theta$.
    
\noindent\textbf{4. User-Specified Informative Parameter:} The parameter \(\eta\) is user-defined, based on prior knowledge or desired level of informativeness.

A key concept in this construction is the selection of the base model. Although this might seem like a rather arbitrary choice, in most situations, there is an obvious option to which one would prefer the prior to contract. For example, suppose that one has a random effect model $y_{ki} = X_{ki}(\beta + \beta_i) + \epsilon_{ki},$ $i=1,\ldots, N_k,$ $k=1,\ldots, K$, where the response $y_{ki}$ denotes the $i$th observation from the $k$th group. Here, $N_k$ is the number of observations in the $k$th group, $K$ is the total number of groups, and $\boldsymbol{X} = \{X_{ki}\}_{k=1, i=1}^{K, N_k}$ represents a covariate that varies across groups and observations. The parameter $\beta$ represents the average effect, while $\beta_k \sim \mathcal{N}(0, \sigma^2)$ models the between-group variations. Finally, $\epsilon_{ki}$ represents the measurement noise, assumed to be independent of $\beta_k$ and $X_{ki}$.
Suppose now that we want to assign a prior to $\sigma^2$ for the random effect. In this case, the simplest model would be that $\sigma = 0$, so that $\beta_k = 0$ and there are no between-group variations. 
This choice is the simplest also from a probabilistic perspective, as it corresponds to a base model which is a Dirac measure concentrated at zero, which is the ``simplest'' possible probability measure. 

Another example is a latent time series model, where $y_i$, $i = 1, \ldots, N$, represents a noisy observation of a Gaussian time series $X(\cdot)$ evaluated at the points $\{t_i\}_{i=1}^n$. The process $X(\cdot)$ could, for instance, be modeled as an autoregressive process, capturing temporal dependencies. 
A natural base model for the prior of the parameters of $X$ would be one so that $X(t) = X$ is a constant function, as there is no point in using the more complex time series if the data could be modelled as independent Gaussian variables with some unknown mean. One could similarly consider a stochastic process or a random field in place of the time series $X(\cdot)$. In such cases, the latent process would be indexed not only by time but potentially by multidimensional spatial or spatiotemporal coordinates.

The four principles allow for systematically constructing priors that avoid overfitting, and the strategy has been shown to provide priors with good properties in several models which are widely used in real applications. For example, \cite{PCpriororigin} propose priors for Gaussian random effects as those in the first example above, for student-t distributions where the base model is chosen as the limiting Gaussian distribution, and for multivariate probit models. \cite{sorbye2017penalised} used the framework to derive priors for autoregressive models as in the second example above. Other important examples where the framework has been used are priors for Gaussian random fields \citep{jasapcprior}, Bayesian P-splines \citep{ventrucci2016penalized}, structured additive distributional regression models \citep{10.1214/15-BA983}, and Weibull models \citep{van2021principled}. 

The intuition behind the choice of KLD as a measure of complexity is that it measures how much information is lost if a flexible model is replaced by the base model. However, there are several issues 
related to this choice that tend to be overlooked. 
The following example illustrates one of the most common issues.
\begin{example}
    \label{kld_example}
    Suppose we aim to construct a PC prior for the precision parameter $\tau = 1/\sigma^2$ of a centered Gaussian distribution $\mathcal{N}(0, \sigma^2)$. Following \citet{PCpriororigin} and based on the same reasoning as above (i.e., that the simplest model in the class is a constant model), the base model $\mu_b$ is chosen as $\mathcal{N}(0, 1/\tau_b)$ with $\tau_b = \infty$, corresponding to a Dirac measure concentrated at 0. Then 
    $$
    \text{KLD}(\mu_{\tau} \| \mu_{\tau_b}) = \frac{\tau_b}{2\tau} \left(1 + \frac{\tau}{\tau_b} \ln \left(\frac{\tau}{\tau_b}\right) - \frac{\tau}{\tau_b}\right).
    $$
    This shows that the KLD is infinite for any $0 < \tau < \infty$, making it impossible to penalize the ``distance'' using KLD. As a result, Principle 2 cannot be applied.
\end{example}
To overcome this, we introduce Principle 5 as an alternative to Principle 2.

\noindent\textbf{5. Alternative complexity measurement:} If $\text{KLD}(\mu_{\theta}, \mu_{\theta_b}) = \infty$ and \(\sqrt{2\text{KLD}(\mu_{\theta}, \mu_{\theta_b + \epsilon})}\) can be written as \(f(\theta, \epsilon) g(\epsilon)\) for ${\epsilon>0}$, with \(\lim_{\epsilon \to 0} f(\theta, \epsilon) = h(\theta, \theta_b) < \infty\) is positive and \(\lim_{\epsilon \to 0} g(\epsilon) = \infty\), then complexity is measured via \(d(\theta) = h(\theta, \theta_b)\).

Although Principle 5 is not officially listed as a principle of the PC prior framework, it is used in examples where Principle 2 fails, as discussed in \citet[Appendix A]{PCpriororigin}. Thus, we formalize it here as an additional principle. Later, we will see more examples where Principle 5 must be applied.

\begin{remark}
All distances induced by different decompositions Principle 5 are equivalent. Indeed, if another decomposition \(  \sqrt{ 2\text{KLD}(\mu_{\theta}, \mu_{\theta_b + \epsilon})} = f_1(\theta, \epsilon) g_1(\epsilon) \) exists, then \(f_1(\theta, \epsilon)/f(\theta, \epsilon) \) does not depend on \(\theta\) since \( f_1(\theta, \epsilon)/f(\theta, \epsilon) g_1(\epsilon) = g(\epsilon) \). Thus, there exists a positive constant $C<\infty$ such that 
${\lim_{\epsilon\rightarrow 0}f_1(\theta, \epsilon)/f(\theta, \epsilon) = C}.$
Therefore, penalizing \(h(\theta,\theta_b) \) is equivalent to penalizing \(h_1(\theta,\theta_b) \) because 
$$h_1(\theta, \theta_b) = \lim_{\epsilon\rightarrow 0}f_1(\theta,\epsilon) = \lim_{\epsilon\rightarrow 0} f(\theta,\theta_b) f_1(\theta, \epsilon)/f(\theta, \epsilon) = Ch(\theta,\theta_b).$$ 
\end{remark}

In Example~\ref{kld_example}, one can choose \(f(\tau,\tau_b^{-1}) = \sqrt{\frac1{\tau} \left(1 + \frac{\tau}{\tau_b} \ln \left(\frac{\tau}{\tau_b}\right) - \frac{\tau}{\tau_b}\right)} \) and $g(\tau_b^{-1}) = \sqrt{\tau_b/2}$. Hence, \(d(\tau) = \tau^{-\frac12}  \).
Thus, the resulting prior density for $\tau$ is
$
    \pi(\tau) = \frac{\eta}{2} \tau^{-\frac{3}{2}} \exp\left(-\eta \tau^{-\frac{1}{2}}\right).
$ 
A coincidence is that, as we will see later, this PC prior is actually penalizing the Wasserstein-$2$ distance, which will be introduced later.

Although Example~\ref{kld_example} is simple, it is practically relevant because specifying priors for Gaussian random effects is a common task. It also underscores a broader issue related to infinite KLD, which arises whenever the probability measures $\mu$ and $\mu_{\theta_b}$ are not absolutely continuous with respect to each other \citep[][Equation~1.1]{infinite_kld}.
If $\mu_{\theta_b}$ is indeed a simpler model, it is often not absolutely continuous with respect to $\mu$, which leads to issues with infinite KLD. Because of this, 
multiple papers constructing PC priors, such as \cite{sorbye2017penalised} and \cite{ventrucci2016penalized} have to apply Principle 5 instead of Principle 2. Further, \cite{jasapcprior} cannot directly apply either Principle 2 or Principle 5, and instead have to do other approximations to obtain a PC prior for the parameters of a Gaussian random field with a Mat\'ern covariance function.
Furthermore, it is generally difficult to interpret the meaning of \(h(\theta,\theta_b) \) in Principle 5.

Another issue with infinite KLD is the lack of a unified approach to address it. Instead of Principle 5, other alternative principles could be formulated, leading to different valid definitions of PC priors. \citet{robert2017principled} also pointed out additional challenges, including the need for further development to extend PC priors to the multivariate case. While \citet[Section 6.1]{PCpriororigin} proposed a general idea and simple cases for such extensions, no practical rule for handling general settings was provided.
A final thing to note about the principles is that the choice of base model is treated as being equivalent to the choice of the parameter $\theta$. This is not an issue for univariate priors, but may be problematic for the multivariate priors, as there then might be several values of $\theta$ which result in the same Dirac measure, which means that no single $\theta$  corresponds to the simplest model.  

Because of these issues, we propose a new type of PC priors, the Wasserstein complexity penalization (WCP) priors, which modify the four PC prior principles. Specifically, the first two principles are adjusted, and the need for a principle to address failures of Principle 2, such as Principle 5, is entirely eliminated. First, instead of selecting the base model through a specific choice of $\theta$, we introduce a base probability measure. This emphasizes that the base probability distribution, denoted by $\mu_b$, should be simpler than the other models $\mu_{\theta} \neq \mu_b$, and that the specific value of the base parameter is not necessarily relevant. Second, to avoid the problems with the KLD, the WCP priors use a penalization based on the Wasserstein distance. 
We show that these modifications solve the issues mentioned above, and that the resulting WCP priors are mathematically tractable and truly follow the stated principles, without the need for alternative principles. We also show that the WCP framework facilitates the construction of multivariate priors in a systematic way. Moreover, we provide \texttt{R} implementations for both the analytical WCP priors derived in this work as well as for numerical approximations for general WCP priors with one- and two-dimensional parameters.
%
To illustrate the flexibility of the approach, we derive WCP priors for several models where PC priors previously have been used, and compare the resulting priors. This also shows that the WCP priors are applicable in all cases where PC priors have been used, covering a large set of models which are important in a wide range of applications. 

The outline of the article is as follows. 
Section~\ref{pcw_section} contains a brief review of the Wasserstein distance and the introduction of the WCP priors for models with a single parameter. 
Section~\ref{uni_application_section} presents several applications of these priors. Section~\ref{multi_section} introduces the multivariate WCP priors, for models with multiple parameters, and Section~\ref{sec:app_multi_wcp} presents two applications of these multivariate WCP priors. The paper ends with a discussion in Section~\ref{discussion} followed by six appendices which present further technical details, and all proofs.
All results in the paper are implemented in the \texttt{R} \citep{Rsoftware} package \texttt{WCPprior}, available at \url{https://vpnsctl.github.io/WCPprior/}. The package also contains \texttt{R} and \texttt{stan} \citep{rstan} functions which can be used to implement the WCP priors in \texttt{stan} and \texttt{R-INLA} \citep{lindgren2015software}.

\section{Wasserstein complexity penalization priors}
\label{pcw_section}
The goal of this section is to introduce the WCP prior for $\theta\in\Theta\subset\mathbb{R}$ in a family of probability measures $(\mu_{\theta})_{\theta\in\Theta}$.
As the WCP prior is based on the Wasserstein distance, we begin with a brief review to their definition and main properties.

\subsection{A brief review of Wasserstein distance}\label{reviewwasser}
The Wasserstein distance can be defined in very general settings. The following definition from \citet{optransoldnew} shows how it is defined for probability measures on a metric space $(\mathcal{X},d)$, where $d$ is the metric. We will, to some extent, need this generality as $\mu_\theta$ can be anything from a univariate distribution to a Gaussian measure induced by a Gaussian random field on $\mathbb{R}^d$.
    
    \begin{definition}
        \label{wassersteindis_defi}
        Let $(\mathcal{X}, d)$ be a separable and complete metric space with the Borel $\sigma$-algebra $\mathcal{B}(\mathcal{X})$. The Wasserstein distance of order $p \in [1, \infty)$ (Wasserstein-$p$ distance) between two probability measures $\mu$ and $\nu$ on $\mathcal{X}$ is defined as
        $$
        W_p(\mu, \nu) = \left(\inf_{\pi \in \Pi(\mu, \nu)} \int_{\mathcal{X} \times \mathcal{X}} d(x, y)^p \, d\pi(x, y)\right)^{1/p},
        $$
        where $\Pi(\mu, \nu)$ is the set of  probability measures on $\mathcal{X} \times \mathcal{X}$ with marginals $\mu$ and $\nu$: Any $\pi \in \Pi(\mu, \nu)$ satisfies $\pi(A \times \mathcal{X}) = \mu(A)$ and $\pi(\mathcal{X} \times B) = \nu(B)\,\, \forall A, B \in \mathcal{B}(\mathcal{X})$.        
    \end{definition}
    
    Letting $P_p(\mathcal{X})$ denote the set of probability measures on $\mathcal{X}$ with finite $p$th moment (see Appendix~\ref{app:wasserstein}  for the precise definition), we have that the Wasserstein distance of order $p$, $W_p$, is a metric in $P_p(\mathcal{X})$. If $\mu_b = \delta_x$ is a Dirac measure concentrated on $x\in\mathcal{X}$, then $\mu_b \in P_p(\mathcal{X})$ for every $p\geq 1$. 
    Thus, whenever $\nu \in P_p(\mathcal{X})$, we have that $W_p(\mu_b,\nu)<\infty$. This is important
    since it means that we can always choose Dirac measures as base models for WCP priors, which is a common choice in the PC prior framework as they are the ``simplest'' measures corresponding to constant random variables. 
    Although $W_p(\mu,\nu)$ is generally challenging to compute, there are many cases where it can be computed analytically. See Appendix~\ref{app:wasserstein} for a discussion of this, examples and different expressions for $W_p(\mu,\nu)$.

\subsection{Univariate WCP priors}
\label{intropcw_section}
Suppose we want to assign a prior to the parameter $\theta$ in a model class $M = \{\mu_{\theta} : \theta \in \Theta\}$, where $\Theta = (\theta_{-}, \theta_{+}) \subset \mathbb{R}$ is an open interval. 
Specifying the WCP prior requires defining a base model, which should be the ``simplest'' model in the extended class $\overline{M} = \{\mu_{\theta} : \theta \in \bar{\Theta}\}$, where $\bar{\Theta} \subset [\theta_{-}, \theta_{+}]$. Let $\mu_b$ denote the base measure, and $\Theta_b$ the base parameter set, that is, we have $\mu_{\theta_b} = \mu_b$ for $\theta_b \in \Theta_b$. For simplicity, we assume $\Theta_b$ is unitary, i.e., $\Theta_b = \{\theta_b\}$. 
Define $\Theta_{-} = (\theta_{-}, \theta_b)$ and $\Theta_{+} = (\theta_b, \theta_{+})$, noting that one of these sets may be empty if $\theta_b$ lies at the boundary of the interval.
    We define 
     $$
     W_{p}^-(\theta) = \begin{cases}
        W_p(\mu_{\theta},\mu_b), & \theta \in \bar{\Theta}_-\\
        0, & \theta \in \Theta_{+}
     \end{cases},\quad
     W_p^{+}(\theta) = \begin{cases}
        0, & \theta \in \Theta_{-}\\
        W_p(\mu_{\theta},\mu_b), & \theta \in \Theta_+
     \end{cases},
     $$
    and let $c_- = \sup_{\theta\in\Theta} W_p^{-}(\theta)\geq 0$ and $c_+ = \sup_{\theta\in\Theta} W_p^{+}(\theta)\geq 0$, which can be infinite.
    %
    %
    We are now ready to give the definition of the WCP$_p$ priors. We do this by following principles similar to those of the PC priors, assigning a truncated exponential distribution as the prior of $W_p(\mu_b,\mu_\theta)$, and performing a change of variables. This yields the following definition. 
    \begin{definition}[WCP$_p$ priors]\label{pcwprior_defi}
        Suppose that $M$ satisfies certain weak regularity assumptions (Assumption~\ref{assump1} in Appendix~\ref{app:assumptions}). Then, the WCP$_p$ prior of $\theta$ has density 
        \begin{equation*}
            \pi(\theta) = W_p^{-} \frac{\eta_- e^{-\eta_- W_p^{-}(\theta)}}{1 - e^{-\eta_- c_-}} \left| \frac{\mathrm{d} W_p^{-}(\theta)}{\mathrm{d} \theta} \right| + W_p^{+} \frac{\eta_+ e^{-\eta_+ W_p^{+}(\theta)}}{1 - e^{-\eta_+ c_+}} \left| \frac{\mathrm{d} W_p^{+}(\theta)}{\mathrm{d} \theta} \right|, \quad  \theta \in \Theta,
        \end{equation*}    
        where $\eta_-,\eta_+>0$ are user-specified hyperparameters to control the tail mass and 
        $$
        W_p^{-} = \frac{1- e^{-\eta_- c_-}}{2 - e^{-\eta_- c_-} - e^{-\eta_+ c_+}}, \quad
        W_p^{+} = \frac{1- e^{-\eta_+ c_+}}{2 - e^{-\eta_- c_-} - e^{-\eta_+ c_+}}.
        $$        
    \end{definition}
    
By construction, the WCP prior satisfies the following principles, where we also include principles for choosing the base model and the order of the Wasserstein distance. In the following, $\overline{M}$ is defined as the closure of $M$ in $P_p(\mathcal{X})$, meaning that $\mu \in \overline{M}$ a sequence $(\theta_n)\subset \Theta$ exists such that $W_p(\mu_{\theta_n}, \mu) \to 0$ as $n\to\infty$.
    
\noindent \textbf{1. Contraction towards simpler measures:} The prior favors models that correspond to simpler measures, where simplicity is relative to the base measure \(\mu_b\). The prior penalizes deviations of \(\mu_\theta\) from \(\mu_b\). In cases where a Dirac measure exists in \(\overline{M}\), it must be chosen as the base measure. Notably, the base measure is independent of model parameterization.
    
\noindent \textbf{2. Complexity Measurement via the Wasserstein distance:} The Wasserstein-$p$ distance is used to measure the deviation from $\mu_b$:  \(d_p(\theta) = W_p(\mu_{\theta}, \mu_{b})\), where $p$ must be chosen such that \(d_p(\cdot)\) depends on the parameter of interest and $\overline{M}\subset P_p(\mathcal{X})$. 
    
\noindent \textbf{3. Constant directed rate penalization:} The penalization rates for deviations from the base model is constant. Specifically, the density \(\pi_{d(\theta)}\) satisfies that \({\pi_{d(\theta)}(d + \delta) = r_{\pm}^\delta \pi_{d(\theta)}(d)}\), where \(r_{\pm} \in (0,1)\) are the decay rates, \(r_+\) applies when \(\theta > \theta_{b,+}\) and \(r_-\) when \(\theta < \theta_{b,-}\), and \(\eta_{\pm} = -\log(r_{\pm})\). 
    
\noindent \textbf{4. User-defined informativeness:} The parameters \(\eta_{\pm}\) are user-specified, based on prior knowledge or the desired level of informativeness. 

The main differences between these principles and the original PC prior principles lie in the first two principles. Additionally, the fourth principle has been slightly extended, as the general WCP$_p$ prior introduces two user-specified parameters.

\begin{remark}\label{rem:wcp_oneside}
In general, the WCP$_p$ prior has two user-defined parameters. However, when $\theta_b = \theta_-$ or $\theta_b = \theta_+$, one of the sets $\Theta_-$ or $\Theta_+$ is empty, leaving only one parameter. For instance, if $\Theta_-$ is empty, then by Definition~\ref{pcwprior_defi}, the density of the WCP$_p$ prior for $\theta$ is 
$$
\pi(\theta) = 
\begin{cases}
    \frac{\eta_+ \exp(-\eta_+ W_p^{+}(\theta))}{1-\exp(-\eta_+ c_+)}\left|\frac{\mathrm{d} W_p^{+}(\theta)}{\mathrm{d} \theta}\right| & \text{if } c_+ < \infty,\\
    \eta_+ \exp(-\eta_+ W_p^{+}(\theta))\left|\frac{\mathrm{d} W_p^{+}(\theta)}{\mathrm{d} \theta}\right| & \text{if } c_+ = \infty,
\end{cases}
$$
for $\theta \in \Theta$. If $\theta_b$ is not at the boundary of the parameter space, the case of a single user-specified parameter can be recovered by setting $\eta_+ = \eta_-$.
\end{remark}

    There are three choices that the user needs to make when specifying a WCP prior:
    \begin{enumerate*}
    \item Choose the base model;
    \item Choose the penalty parameter $\eta=\eta_-=\eta_+$ (or $\eta_-$ and $\eta_+$ separately); and
    \item Choose $p$ in the Wasserstein distance.
    \end{enumerate*}
    As previously mentioned, there is typically only one choice of base model if the goal is to penalize complexity in the model. However, in certain cases, there may be more than one plausible choice \citep[see, e.g.,][]{sorbye2017penalised} and in this situation the user needs to decide on which base model that the prior should contract towards as a modelling choice. 
    The choice of the penalty parameter is, by design, application dependent, and the parameter value is chosen based on prior information. For example, $\eta$ can be chosen by specifying the prior probability that $|\theta-\theta_b|>c$ for some user specified $c>0$, which is often something the user may have prior knowledge about \citep[see, e.g.][]{PCpriororigin}.
    For the final choice of $p$, suppose that $\mu_{\theta}$ has finite $k$th moment for $k = 1,2,\ldots, K$. In this case, we must choose $p\leq K$, and $p$ must be chosen such that the Wasserstein distance depends on the parameters of interest, which may enforce a lower bound on $p$. If there are multiple values of $p$ satisfying these requirements, we typically prefer choices that provide simple and closed-form expressions of the prior. Throughout the paper, when the explicit knowledge of the order is not required, we will refer to the WCP$_p$ priors simply as the WCP priors.

    The problem of infinite KLD in the original PC prior, mentioned in the introduction, which makes it necessary to consider the alternative PC prior principles (Principles 1, 3, 4 and 5) is completely avoided in the WCP priors since $W_p$ is finite on $P_p(\mathcal{X})$. 
    Further, since $W_p$ metrizes the weak convergence of 
    probability measures in $P_p(\mathcal{X})$, the interpretability is enhanced, as we can describe the shrinkage towards the base model precisely. Finally, 
    an important feature of the $\text{WCP}_p$ priors is that they are invariant under reparameterization in the sense that the principles are still obeyed under reparameterization. More precisely, we have the following proposition, which follows directly from the definition of the WCP prior and the chain rule.
    \begin{proposition}\label{prp:invariancerepar}
        Let $g:\Theta\to (\phi_-, \phi_+)$ be an invertible and differentiable function with nonvanishing derivative. Let $\phi = g(\theta)$ be a reparameterization of the model in Definition \ref{pcwprior_defi}. 
        If $\pi(\theta)$ and $\pi(\phi)$ are the $\text{WCP}_p$ prior densities for $\theta$ and $\phi$, respectively, then 
        $\pi(\theta) = \pi(\phi) |g'(\theta)|$,
        where ${\phi = g(\theta)}$. Thus, the $\text{WCP}_p$ prior for $\phi$ is obtained by applying the change of variables ${\phi = g(\theta)}$ on the $\text{WCP}_p$ prior of $\theta$.
    \end{proposition}
    
\section{Applications of WCP priors}
\label{uni_application_section}
\subsection{A class of location-scale models}
Let \(\boldsymbol{X}\) be a random variable taking values in \(\mathbb{R}^d\), and consider the family of distributions given by \(\{\mu_{\boldsymbol{m},\sigma}: \boldsymbol{m} \in \mathbb{R}^d, \sigma > 0\}\), where \(\mu_{\boldsymbol{m},\sigma}\) is the distribution of \(\sigma \boldsymbol{X} + \boldsymbol{m}\). Given $\boldsymbol{s} \in\mathbb{R}^d$, the Wasserstein distance between \(\mu_{\boldsymbol{m},\sigma}\) and a Dirac measure \(\delta_{\boldsymbol{s}}\) can then be easily computed in terms of the moments of \(\boldsymbol{X}\) using Proposition~\ref{loc-scale-prop}. Several families of distributions belong to this class, such as Gaussian, 
exponential, logistic, half-normal, Maxwell, 
Rayleigh, etc. In this case, \(\{\mu_{\boldsymbol{m},\sigma}: \boldsymbol{m} \in \mathbb{R}^d, \sigma > 0\}\) constitutes a location-scale family of distributions which we will referred to as location-scale distributions generated by $\boldsymbol{X}$.
Let us revisit  Example~\ref{kld_example} from the introduction and derive the corresponding WCP prior for a broader class of distributions belonging to this family of distributions.
    
    \begin{proposition}\label{example1}
        Fix a random variable \( \boldsymbol{X} \) in \(\mathbb{R}^d\) with a finite \( p \)-th moment, where \( p \geq 1 \). Let \(\mu_\sigma\) denote the distribution of \(\sigma \boldsymbol{X}\), and let \(\mu_b = \delta_{\boldsymbol{0}}\) be the base measure corresponding to \(\sigma = 0\). The \(\text{WCP}_p\) prior for \(\sigma\) is then an exponential distribution with density
        $
        \pi_p(\sigma) = \eta C_p \exp(-\eta C_p \sigma),
        $
        where \( C_p = (\mathbb{E}\|\boldsymbol{X}\|_{\mathbb{R}^d}^p)^{1/p} \). Furthermore, by incorporating the finite constant $C_p$ into the user-specified parameter \(\eta\), the prior is independent of $p$ and the \(\text{WCP}_p\) prior density for \(\tau = 1/\sigma^2\) is a Type-2 Gumbel distribution with density 
        $
        {\pi_p(\tau) = \frac{1}{2} \tau^{-3/2} \eta \exp(-\eta \tau^{-1/2}).}
        $        
    \end{proposition}
    
    This result follows directly from Definition \ref{pcwprior_defi}, as by Proposition~\ref{loc-scale-prop}, \(W_p(\mu_\sigma, \mu_b) = \sigma C_p\).
    In Proposition~\ref{example1}, the Wasserstein distance between the base model and the flexible model is finite, and when $\tau \rightarrow \infty$, $\mu_\tau$  converges weakly to $\mu_b$ in $P_p(\mathbb{R}^d)$. 
    On the other hand, for any $\boldsymbol{X}$ such that $P(\boldsymbol{X} = \boldsymbol{0}) < 1$, the KLD in the original PC prior is infinite for all $0 < \tau < \infty$. By using Principle 5 in the original PC prior in place of Principle 2, the prior becomes difficult to interpret in terms of penalization with respect to KLD.
    An interesting coincidence is that if we take $\boldsymbol{X}$ following a standard normal distribution the WCP prior for \(\tau\) that we discussed in Example~\ref{kld_example} is identical to the PC prior from \citet[][Appendix A.1]{PCpriororigin}. This indicates that the  PC prior based on Principles 1,3,4 and 5 has a meaningful interpretation in terms of the Wasserstein distance.
    
    Similarly, we have the following result regarding the WCP prior for the location parameter, which follows directly from Remark~\ref{remarkwformula} and Definition \ref{pcwprior_defi}. This is a case where $\Theta_b = \{\theta_b\}$ and $\theta_b$ is in the interior of $\Theta$.
    
    \begin{proposition}
        \label{example2}
        Let \( X \) be a random variable in \(\mathbb{R}\) with a finite \( p \)-th moment, where \( p \geq 1 \). Define \(\mu_m\) as the distribution of \(\sigma X + m\), where \(\sigma^2 < \infty\) and $m\in\mathbb{R}$. Take the base measure \(\mu_b\) as the distribution of \(\sigma X\). The \(\text{WCP}_p\) prior density for \(m\) is
        \begin{equation}\label{eq:mean_prior}
            \pi_p(m) = \begin{cases}
                \frac{\eta_+}{2} \exp(-\eta_+ m) & \text{for } m > 0, \\
                \frac{\eta_-}{2} \exp(\eta_- m) & \text{for } m < 0.
            \end{cases} 
        \end{equation}
        Setting \(\eta_- = \eta_+ = \eta\) yields the Laplace prior 
            $\pi_p(m) = \frac{\eta}{2} \exp(-\eta |m|)$.
    \end{proposition}

	\subsection{Stationary autoregressive processes}
    \label{AR1pcw_section}
    Time series models are important in a number of applications, and whenever they are included in Bayesian models, there is a need to design priors for their parameters.
    In this section, we discuss WCP priors for weakly stationary auto-regressive processes of order 1, denoted by AR(1), which arguably is one of the most popular time series models \citep{chi1989models, jones1991unequally, rue2005gaussian, prado2010time}.  
    A centered weakly stationary AR(1) process is a discrete-time stochastic process, $\{X_t,t\in \mathbbm{N}\}$, defined by the recursive relation $X_0 \sim \mathcal{N}(0,\sigma^2)$ and $X_t = \phi X_{t-1} + \varepsilon_t$, for ${t = 1,\ldots,n}$,
    where $\abs{\phi}\leq 1$ and $\{\varepsilon_t\}_{t = 1}^n$ are i.i.d Gaussian $\mathcal{N}(0,\sigma^2(1-\phi^2))$, with $\varepsilon_t$ 
    being independent from $X_0$ for each $t \in \{1,2,...,n\}$. 
    In this parameterization, the marginal variances of the process are  independent of $\phi$ and the process is stationary even for $\abs{\phi} = 1$.
    
    \citet{sorbye2017penalised} proposed a PC prior for the parameter $\phi$. 
    They considered a base model for $\phi$ with two cases; either the probability distribution corresponding to $\phi_0 = 0$,  which is discrete-time white noise, or the probability distribution corresponding to $\phi_0 = 1$, which is constant in time. In both cases we have $\Theta_b = \{\phi_0\}$. 
    Here, we focus on the case \(\phi_0 = 1\) as the base model, since a constant process is inherently ``simpler'' than an AR(1) process and a white noise. In this scenario, the KLD between the base model and a flexible model with \(\phi < 1\) is infinite \citep[][Section~3.2]{sorbye2017penalised}, necessitating the usage of Principle 5 in place of Principle~2. We will now demonstrate that this is not the case for the WCP$_2$ priors.
    
    Let $\mu_{b}$ denote the Gaussian measure induced by the base model (with $\phi=1$) and $\mu_\phi$ be the measure induced by the flexible model with $\phi \in [-1,1)$. The WCP$_2$ prior for $\phi$ is given in the next proposition, whose proof is provided in Appendix~\ref{app:proofs}.

    \begin{proposition}\label{prp:AR_wcp_density}
        The $\text{WCP}_2$ prior density for $\phi$ of the centered AR(1) process, with respect to the base model $\mu_{b}$ with $\phi=1$ is given by
        $$
        {\pi(\phi) = \sigma\left|\frac{(n\phi^n-1+\phi^n-n\phi)(1-\phi) + f_n(\phi)^2}{\sqrt{2}f_n(\phi)\sqrt{n - \frac{f_n(\phi)}{1-\phi}}(1-\phi)^2} \right|\frac{\eta \exp\Bigl(-2\eta \sigma^2\Bigl(n - \frac{f_n(\phi)}{1-\phi}\Bigr)\Bigr)}{1-\exp(-\eta c)}},
        $$ 
        where 
        ${f_n(\phi) = \sqrt{n(1-\phi^2)-2\phi(1-\phi^n)}}$ and $c = \sigma\bigl(2n - \sqrt{2}\sqrt{1 - (-1)^n}\bigr)^{1/2}$.
    \end{proposition}
    The \(\text{WCP}_2\) prior differs from the PC prior presented in \citet[Section 3.2, Equation 6]{sorbye2017penalised}. Specifically, the \(\text{WCP}_2\) prior depends on the length \(n\) of the AR(1) process by obeying its four principles, whereas the PC prior manages to be independent of \(n\) by following Principles 1,3,4,5 and absorbing \(n\) into the user-specified parameter.  
    Nevertheless, being dependent of \(n\) is not a flaw of the \(\text{WCP}_2\) prior but a byproduct of its principles.
    The dependence of the \(\text{WCP}_2\) prior on \(n\) is natural if we can observe the full AR(1) process: the complexity of the base model (which is constant over time) is independent of \(n\), while the complexity of the AR(1) process increases with \(n\). Consequently, the prior reflects this dependence on \(n\). However, one drawback with the  \(\text{WCP}_2\) prior is that its density converges to a uniform distribution as $n\rightarrow\infty$, so it cannot be used for constructing a meaningful prior for an AR(1) process on $\mathbb{N}$. The reason being that the constant base model is not in $P_2(\mathbb{R}^\mathbb{N})$ and thus not in the domain of the Wasserstein-2 distance.


    \begin{figure}[t]
        \centering
        \includegraphics[scale = 0.95]{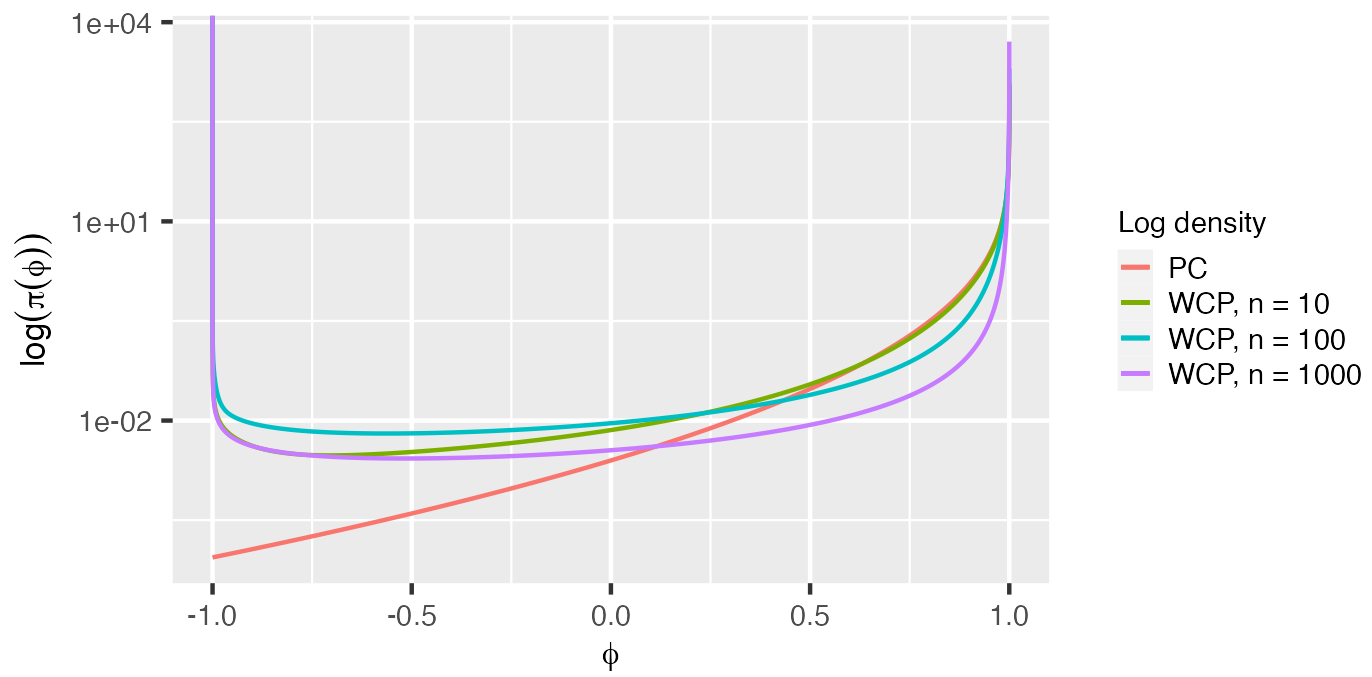}
        \caption{PC and $\text{WCP}_2$ prior densities in log scale. The user-specified hyperparameter for the $\text{PC}$ prior is $\theta \approx 7.28$.
            The corresponding parameter for the $\text{WCP}_2$ prior is $\eta \approx 13.44$ when $n = 10$, $\eta \approx 2.17$ when $n = 100$, and $\eta \approx 0.57$ when $n = 1000$.}   
        \label{fig:NCTplot}  
    \end{figure}
    A comparison between the two priors is shown in Figure \ref{fig:NCTplot}. The user-specified hyperparameters for both priors are chosen so that they satisfy $P(\phi > 0.9) = 0.9$, and $\sigma = 0.1$.
    The $\text{WCP}_2$ prior assigns less mass near the base model $\phi = 1$ than the $\text{PC}$ prior for $n=10$, and when $n$ increases, the $\text{WCP}_2$ prior becomes more concentrated around the base model.
    To further compare the priors, we performed a simulation study that compares the Maximum A Posteriori (MAP) estimations under the $\text{WCP}_2$ prior, the PC prior, and a uniform prior on $\phi$ with simulated data from an AR(1) process with $n = 10, 100, 1000$. For each value of $n$, we generated data with $\phi = 0.5$ and then computed MAP estimates. This was repeated $5000$ times, and the whole procedure was then repeated with data where the true parameter is $\phi = - 0.5$. 
    Figure~\ref{fig:sim_AR_02} shows box plots of the resulting estimates. 
    Compared to the uniform prior, the MAP estimations under both the $\text{WCP}_2$ and the PC priors are biased towards the base model when $n$ is small, while when $n$ is larger, that bias disappears.
    This is reasonable because a small value of $n$ does not provide strong evidence against the base model.
    However, as expected from the results in Figure~\ref{fig:NCTplot}, the $\text{WCP}_2$ prior has a slightly lower bias for small values of $n$, even though the user-specified parameters are chosen in the same way.  
    
    \begin{figure}[t]
        \centering
        \includegraphics[width = 0.8\linewidth]{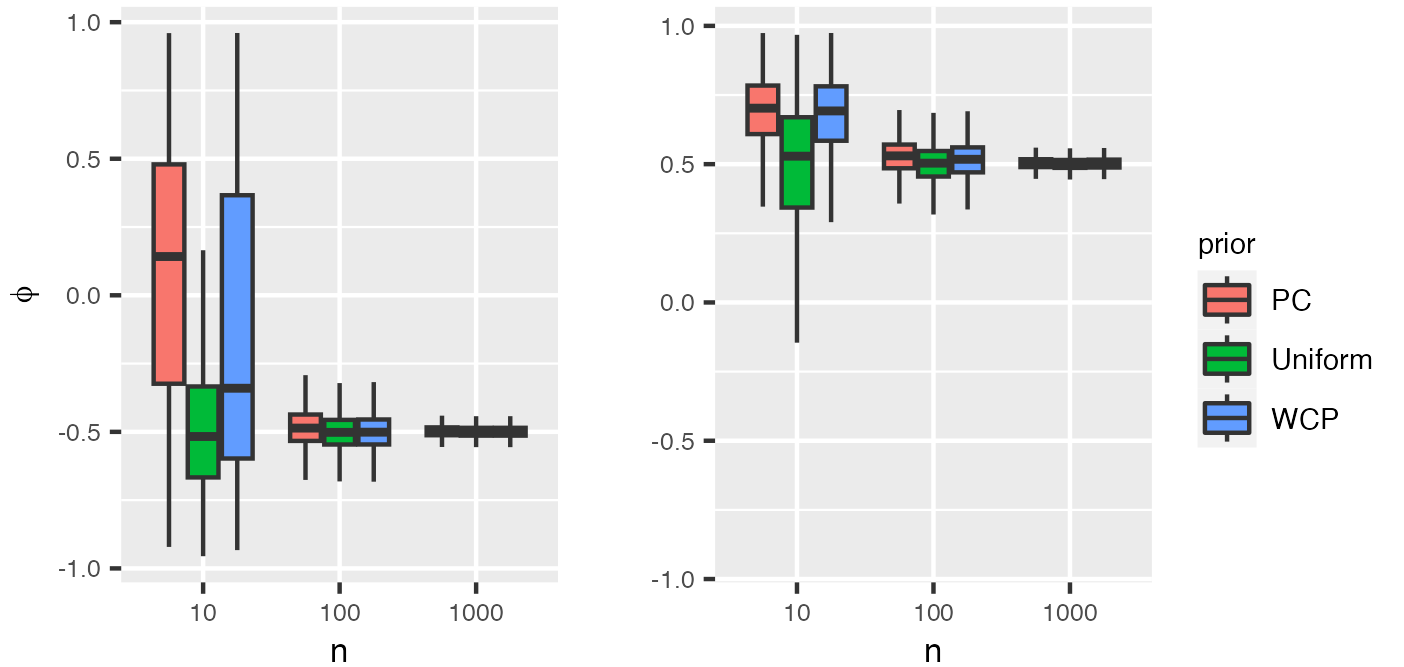}
        \caption{Distributions of MAP estimations of $\phi$ under different priors, where data is generated with $\phi = -0.5$ (left) and $\phi = 0.5$ (right).}     
        \label{fig:sim_AR_02}
    \end{figure}
    
    
    \subsection{Extreme value statistics models}
    \label{sec:tailGP}
    Extreme value statistics is an important branch of statistics, concerned with the study of extreme events. Even though traditional extreme value statistics was mainly done in a frequentist setting, it is now frequently done in a Bayesian context \citep{bousquet2021extreme, dombry2017bayesian, de2022extreme}. It is then important to design priors for the parameters, and in this section we consider this problem for one of the classical models in extreme value theory. Specifically, we will derive a prior
    for the tail index $\xi\in \mathbbm{R}$ of a generalized Pareto (GP) distribution with density
    ${f_\xi(y) = \sigma^{-1}(1+\xi y/\sigma)^{-1/\xi-1},y>0}$, where $\sigma>0$ is a scale parameter. 

    \citet{opitz2018inla} considered $\xi\in [0,1)$ since other values of $\xi$ are not realistic for modeling, and derived a PC prior for $\xi$. 
    When $\xi\in [0,1)$, the GP distribution has a finite first moment, and the associated probability measures thus belong to $P_1(\mathbb{R})$.  
    Therefore, it is natural to consider the \(\text{WCP}_1\) prior for \(\xi\), since \(W_p\) may be infinite if \(p > 1\). For instance, the GP distribution only has a finite second moment when \(\xi < 0.5\). 
    When $\xi = 0$, the GP distribution is the exponential distribution which has the lightest tail compared to other values of $\xi \in [0,1)$. Therefore, \citet{opitz2018inla} chose $\mu_b$ as the exponential distribution with density $\pi_{\xi_b}(y) = \sigma^{-1}\text{exp}(-y/\sigma), \sigma>0$. 
    
    \begin{proposition}\label{prp:gen_pareto_wcp_density}
        The $\text{WCP}_1$ prior for $\xi$ with respect to the base model induced by ${\xi_b=0}$, is
        $$
        \pi(\xi) = \frac{\eta}{(1-\xi)^2}\exp\Bigl(-\eta\frac{\xi}{1-\xi}\Bigr), \quad \xi\in (0,1),
        $$
        where $\eta>0$ is the user-specified hyperparameter controlling the tail mass.
        \begin{figure}[t]
        \centering         
        \includegraphics[scale = 1]{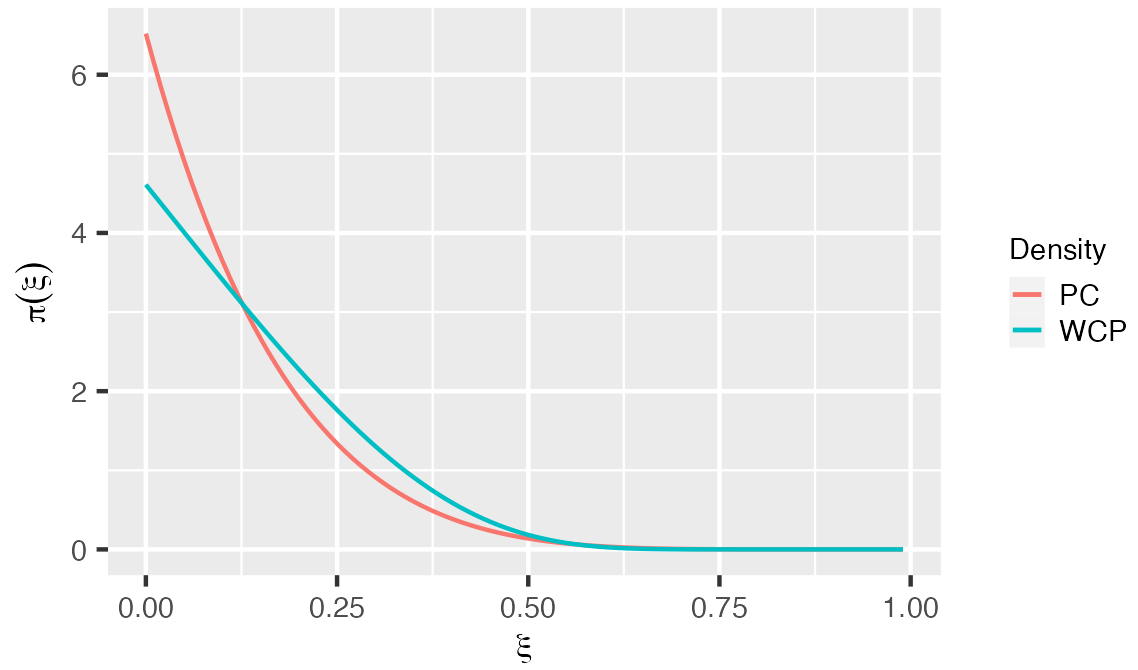}
        \caption{PC and $\text{WCP}_1$ prior densities of the tail index of the generalized Pareto distribution. The user-specified parameter for both priors is 4.61.}
        \label{fig:extreme_plot}
    \end{figure}
\end{proposition}
    
    Figure \ref{fig:extreme_plot} shows the $\text{PC}$ prior from \citet[eq. 8]{opitz2018inla} and the $\text{WCP}_1$ prior. 
    Both priors satisfy $P(\xi >0.5) = 0.01$. 
    The $\text{PC}$ prior concentrates more than the $\text{WCP}_1$ prior around the base model, which means that it has a higher penalization.



\section{Multivariate WCP priors}
\label{multi_section}
In this section, we extend the WCP priors to models with multiple parameters. 
Before introducing the WCP priors, we revisit the choice of base model which is a bit more delicate in the multivariate setting as we mentioned in the introduction. 

\subsection{Notation and preliminaries}
\label{sec:base_measure_concept}

Suppose we want to assign a prior to a parameter $\boldsymbol{\theta}\in\boldsymbol{\Theta} \subset \mathbb{R}^d$, $d \in \mathbb{N}$, for a model class  $M = \{\mu_{\boldsymbol{\theta}}: \boldsymbol{\theta} \in \boldsymbol{\Theta}\} \subset P_p(\mathcal{X})$.
Specifying the WCP$_p$ prior requires defining a base model, which should be the ``simplest'' model within an extended model class $\overline{M}$. A distribution $\mu$ belongs to $\overline{M}$ if and only if it can be approximated by models in $M$ with respect to the Wasserstein distance (see Appendix~\ref{app:basemeasure} for details).
If we have a base model $\mu_b \in \overline{M}$, we can find \(\boldsymbol{\theta}_b \in \overline{\boldsymbol{\Theta}}\) such that there exists a sequence \(\boldsymbol{\theta}_n \to \boldsymbol{\theta}_b\) and \(W_p(\mu_{\boldsymbol{\theta}_n}, \mu_{\boldsymbol{\theta}_b}) \to 0\) as \(n \to \infty\), where \(\overline{\boldsymbol{\Theta}}\) is the closure of \(\boldsymbol{\Theta}\) in $\overline{\mathbb{R}}^d$, with \(\overline{\mathbb{R}} = [- \infty, \infty]\) being the extended line. 
This allows us to define \(\mu_{\boldsymbol{\theta}_b}\) for some values of \(\boldsymbol{\theta}_b\) that are not in \(\boldsymbol{\Theta}\) but in \(\overline{\boldsymbol{\Theta}}\). Let \(\boldsymbol{\Theta}_b \subset \overline{\boldsymbol{\Theta}}\) be the set of all such parameters. We refer to this set as the base parameter set. Observe that this enables the identification 
$\overline{M} := \{\mu_{\boldsymbol{\theta}}: \boldsymbol{\theta} \in \boldsymbol{\Theta} \cup \boldsymbol{\Theta}_b\}.$
We will always assume that \(\boldsymbol{\Theta}_b \subset \mathbb{R}^d\) is a connected set, as this ensures monotonicity when moving ``away'' from the base model.

As we stated for the univariate WCP priors, whenever \(\overline{M}\) contains Dirac measures, one of these must be chosen as the base model, as they represent the simplest measures possible. 
Further, it is important to note that \(\boldsymbol{\Theta}_b\) does not need to be unitary, which is particularly relevant when \(\boldsymbol{\Theta} \subset \mathbb{R}^d\) for \(d > 1\). For example, for a model with parameters \((\sigma, \theta_2, \ldots, \theta_{d-1})\), 
$\{\boldsymbol{\theta} = (\sigma, \theta_2, \ldots, \theta_{d-1}) \in \overline{\boldsymbol{\Theta}}: \sigma = 0\} \subset \boldsymbol{\Theta}_b$. This is the main reason for considering a base measure instead of a base model with a fixed parameter.


\subsection{Definition and properties}

We will first present an informal and intuitive definition of the multivariate WCP$_p$ prior in Definition \ref{def:informal_multi}, followed by a rigorous definition in Definition \ref{multi_pcwprior_defi}.

Let $\mu_b$ be the base measure and suppose that the base parameter set $\boldsymbol{\Theta}_b$ is connected. 
Let \(W_p(\boldsymbol{\theta}) = W_p(\mu_b, \mu_{\boldsymbol{\theta}})\) denote the Wasserstein-\(p\) distance between \(\mu_{b}\) and a flexible model \(\mu_{\boldsymbol{\theta}}\) and we define \(\sup_{\boldsymbol{\theta} \in \boldsymbol{\Theta}} W_p(\boldsymbol{\theta}) = c\), where $c=+\infty$ is allowed. 
As in the univariate setting, we need a few weak regularity conditions to guarantee that the prior is well-defined. These are provided in Assumption~\ref{multi_pwc_assumptions} in Appendix~\ref{app:assumptions}.
%

\begin{definition}[Multivariate WCP priors informally]\label{def:informal_multi}
    A multivariate WCP prior for $\boldsymbol{\theta}$ is constructed by assigning a (possibly truncated) exponential distribution to the Wasserstein distance $w = W_p(\boldsymbol{\theta})$. Given $w$, a uniform distribution is assigned over the level set $S_{w,\boldsymbol{\theta}} = \{\boldsymbol{\theta} \in \overline{\boldsymbol{\Theta}} \mid W_p(\boldsymbol{\theta}) = w\}$. Thus, a complexity penalty is imposed based on $w = W_p(\boldsymbol{\theta})$, treating all models with the same $w$ equivalently. The prior for $\boldsymbol{\theta}$ is then derived through a change of variables, analogous to the univariate case.
\end{definition}
\citet[Section~6.1]{PCpriororigin} proposed a similar approach using KLD. They derived multivariate PC priors for a restricted class of models with specific forms of KLD; however, no examples were provided for more general forms and it is not common to find Wasserstein distances that satisfy the requirements in \citet[Section~6.1]{PCpriororigin} under natural model parameterizations.

Let us now move to the formal definition of the multivariate WCP priors. To facilitate the presentation, we will assume that for every $w$, $S_{w,\boldsymbol{\theta}}$ is compact, and that there exists a parameterization $X_w: U_w \subset \mathbb{R}^{d-1} \to \widetilde{S}_{w,\boldsymbol{\theta}} \subset S_{w,\boldsymbol{\theta}}$, such that $\text{Area}_{d-1}(S_{w,\boldsymbol{\theta}}\setminus \widetilde{S}_{w,\boldsymbol{\theta}}) = 0$, where $\text{Area}_{d-1}(\cdot)$ stands for the $(d-1)$-dimensional surface area, see Appendix~\ref{app:assumptions} for more details.
In the definition, $J_{\boldsymbol{g}}(\boldsymbol{x})$ denotes the jacobian matrix of a differentiable function $\boldsymbol{g}$ evaluated at $\boldsymbol{x}$.

\begin{definition}[Multivariate WCP priors]\label{multi_pcwprior_defi}
    Suppose that $M$ satisfies certain weak regularity assumptions (Assumption~\ref{multi_pwc_assumptions} in Appendix~\ref{app:assumptions}). Additionally, for each $w>0$ where $S_{w,\boldsymbol{\theta}}\neq \emptyset$, suppose that the map $(w, \boldsymbol{u}) \mapsto X_w(\boldsymbol{u})$ is a local diffeomorphism. Let $\boldsymbol{u} = (u_1,\ldots,u_{d-1})$ represent the parameters of $X_w$. Then, the WCP prior density of $w$ and $\boldsymbol{u}$ is 
    \begin{equation}
        \label{multi_pcwprior_density_no_change_of_variables}
        \pi(w,\boldsymbol{u}) = \frac{\eta\exp(-\eta w)}{1-\exp(-\eta c)}\frac{\sqrt{\det J_{X_w}(\boldsymbol{u})^\top J_{X_w}(\boldsymbol{u})}}{\text{Area}_{d-1}(S_{w,\boldsymbol{\theta}})},
    \end{equation}
    where $J_{X_w}(\boldsymbol{u})$ has size $d \times (d-1)$ and $\eta > 0$ is a hyperparameter. Now, let $\boldsymbol{\Phi}: \boldsymbol{\Theta} \to \mathbb{R}^{d}$ be the map  $\boldsymbol{\Phi}(\boldsymbol{\theta}) = (W_p(\boldsymbol{\theta}), X_{W_p(\boldsymbol{\theta})}^{-1}(\boldsymbol{\theta}))$. Then, by the change of variables induced by $\boldsymbol{\Phi}$ in \eqref{multi_pcwprior_density_no_change_of_variables}, we arrive at the WCP$_p$ prior density of $\boldsymbol{\theta}$:
    \begin{equation}
        \label{multi_pcwprior_density}
        \pi(\boldsymbol{\theta}) = |\det J_{\boldsymbol{\Phi}}(\boldsymbol{\theta})| \frac{\eta\exp(-\eta W_p(\boldsymbol{\theta}))}{1-\exp(-\eta c)}\frac{\sqrt{\det \boldsymbol{G}(\boldsymbol{\theta})}}{\text{Area}_{d-1}(S_{W_p(\boldsymbol{\theta}),\boldsymbol{\theta}})},
    \end{equation}
    where $\boldsymbol{G}(\boldsymbol{\theta}) = J_{X_{W_p(\boldsymbol{\theta})}}(X_{W_p(\boldsymbol{\theta})}^{-1}(\boldsymbol{\theta}))^\top J_{X_{W_p(\boldsymbol{\theta})}}(X_{W_p(\boldsymbol{\theta})}^{-1}(\boldsymbol{\theta}))$ and 
    $J_{X_{W_p(\boldsymbol{\theta})}}(X_{W_p(\boldsymbol{\theta})}^{-1}(\boldsymbol{\theta}))$ is $J_{X_w}(\boldsymbol{u})$ evaluated at $w=W_p(\boldsymbol{\theta})$ and $\boldsymbol{u} = X_{W_p(\boldsymbol{\theta})}^{-1}(\boldsymbol{\theta})$.
\end{definition}

We refer the reader to Appendix~\ref{app:assumptions} for the most general definition of the multivariate WCP prior, which allows for more general forms of level sets. 

\begin{remark}\label{remark:diffeomorphism}
    The map $(w, \boldsymbol{u}) \mapsto X_w(\boldsymbol{u})$ is a local diffeomorphism if the following set defined as ${O =\{(w,\boldsymbol{u}): w>0, \boldsymbol{u} \in U_w\}}$ is open in $\mathbb{R}^d$ and the map is continuously differentiable and has a non-zero Jacobian determinant for all $(w, \boldsymbol{u})\in O$. Observe that the inverse of this map is given by $\boldsymbol{\Phi}$ so it is enough to check that $\boldsymbol{\Phi}$ is a local diffeomorphism.
\end{remark}

As an example, the following result (derived in Appendix~\ref{app:proofs}) shows the bivariate $\text{WCP}_2$ prior for the mean and the standard deviation of a Gaussian distribution. 

\begin{proposition}
    \label{2d_gaussian_wcp}
    Let $\mu_{\boldsymbol{\theta}} = \mathcal{N}(m,\sigma^2)$ for $\boldsymbol{\theta} = (m, \sigma) \in \mathbb{R} \times (0,\infty)$. 
    Then, the $\text{WCP}_2$ prior of $(m,\sigma)$ is has density
    \begin{equation}
        \label{2dGaussianWCPdensity}
        \pi(m,\sigma) = \frac{\eta\exp(-\eta (m^2+\sigma^2)^{1/2})}{\pi (m^2+\sigma^2)^{1/2}}.
    \end{equation}
\end{proposition}

This proposition is derived directly from the definition of the WCP prior, as the arc lengths of the level curves (the terms $\text{Area}_1(S_{w,(\mu,\sigma)}), w>0,$ in this context) are known in closed form. An example of this prior is shown in the left panel of Figure~\ref{fig:Gaussian_2d_prior}.

In general, computing the WCP prior according to Definition~\ref{multi_pcwprior_defi} requires knowledge of the surface areas of the level sets. However, in Appendix~\ref{app:recipe}, we provide a  recipe for computing WCP priors when the surface areas are unknown. 

\subsection{The two-step approach} \label{sec:multi_prior_Gaussian}

\citet[Section 2.2]{jasapcprior} proposed a way to derive a joint PC prior of two parameters, which is commonly used in practice and which we will refer to as the two-step approach.
In this section, we will formalize a counterpart of this idea for WCP priors and compare it with the true multivariate WCP priors.

Suppose that we have two parameters $\theta_1, \theta_2$, and that the base parameter set is unitary, $\boldsymbol{\Theta}_b = \{(\theta_{1,b}, \theta_{2,b})\}$.
The first step of the two-step approach is to derive a WCP prior for one of the parameters, say $\theta_1$, while fixing $\theta_2 = \theta_{2,b}$.
That is, this WCP prior penalizes the distance between $\mu_{\theta_{1,b}, \theta_{2,b}}$ and $\mu_{\theta_{1}, \theta_{2,b}}$.
This prior is a conditional distribution of $\theta_1$ given that $\theta_2 = \theta_{2,b}$.
However, in the two-step approach,  this is treated as a prior of $\theta_1$, and is denoted by $\pi(\theta_1)$. 
The second step is to derive the conditional WCP prior $\pi(\theta_2|\theta_1)$ for $\theta_2$ given $\theta_1$, that is, the prior penalizes the distance between $\mu_{\theta_{1}, \theta_{2,b}}$ and $\mu_{\theta_{1}, \theta_{2}}$ where $\mu_{\theta_{1}, \theta_{2,b}}$ is considered as the base model. 
The two-step WCP prior density is then $\pi(\theta_1)\pi(\theta_2|\theta_1)$.
\begin{example}
	\label{2d_Gaussian_twostep}
	Let us derive the two-step approach prior for $m$ and $\sigma$ of a $\mathcal{N}(m,\sigma^2)$ distribution with $\mu_b = \delta_{(0,0)}$ as base measure. 
	We derive the $\text{WCP}_2$ prior for $m$ with $\sigma = 0$ first.
    Because $W_2(\mathcal{N}(m,0),\mu_b) = \abs{m}$, we have 
	$\pi(m|\sigma = 0) = \eta_1\exp(-\eta_1 \abs{m})$ for $m \neq 0$,
	where $\eta_1$ is a user-specified hyperparameter.
	Next, we have that $W_2(\mathcal{N}(m,\sigma),\mathcal{N}(m,0)) = \sigma$.
	Therefore, for $\sigma > 0$,
	$\pi(\sigma|m) = \eta_2\exp(-\eta_2 \sigma)$ for $\sigma > 0$,
	where $\eta_2$ is a user-specified hyperparameter.
	Combining the two steps yields the two-step prior
	$
		\pi(m,\sigma) = \frac12 \eta_1\eta_2\exp(-\eta_1\abs{m} - \eta_2 \sigma).
	$
	Figure~\ref{fig:Gaussian_2d_prior} shows the $\text{WCP}_2$  prior from \eqref{2dGaussianWCPdensity} and the two-step prior. We can note that the two priors behave very differently.
    \begin{figure}[t]
        \centering
        \includegraphics[width=0.8\linewidth]{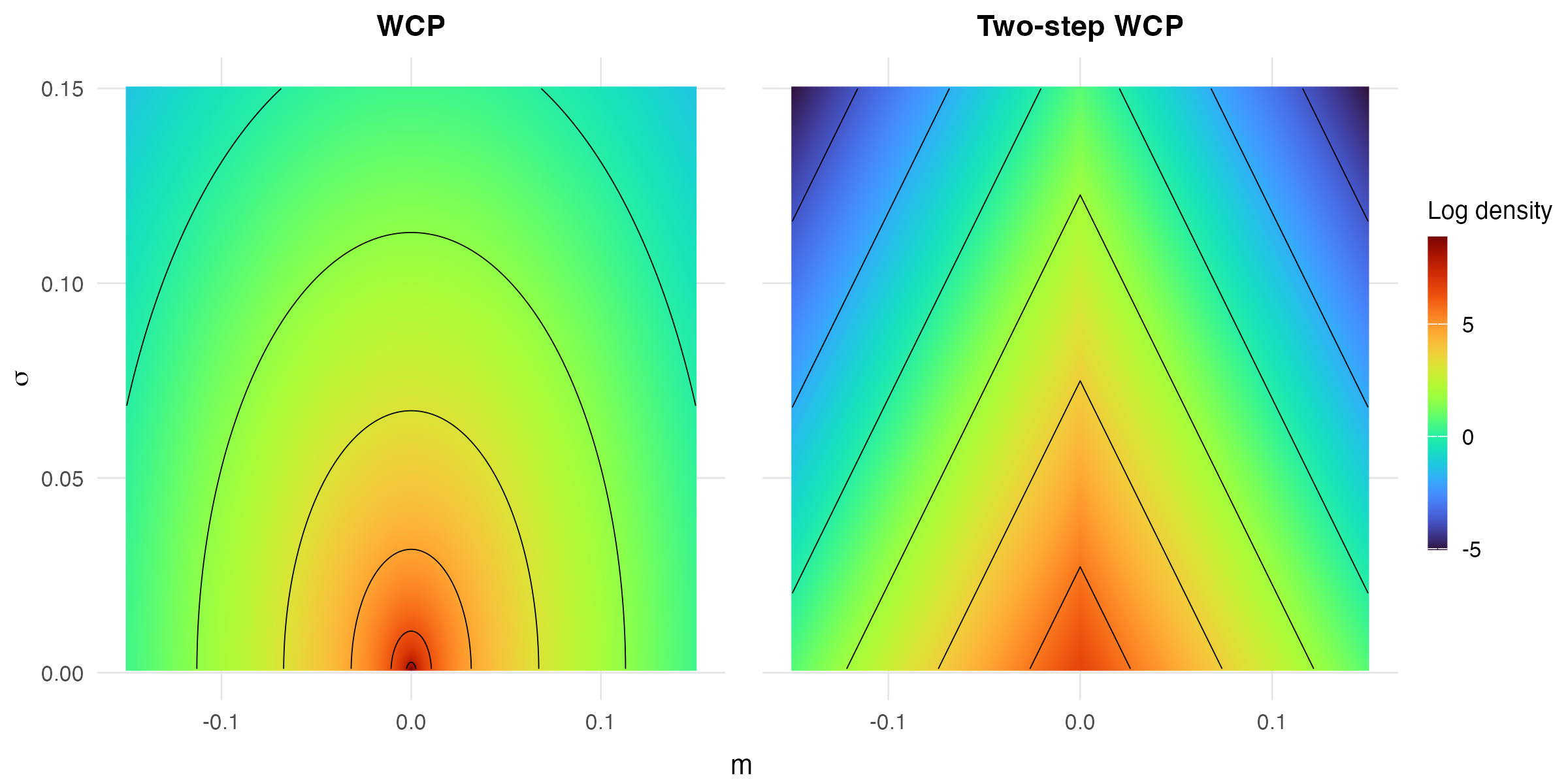}
        \caption{$\text{WCP}_2$ prior (left) with $\eta = 23$ and two-step approach prior (right) with $\eta_1=\eta_2=38.9$  for $(m,\sigma)$ for the normal distribution. The hyperparameters were chosen in such a way that for WCP prior, $P(\|(m,\sigma)\|_{\mathbb{R}^2} \geq 0.1) \approx 0.1$ and for the two-step WCP prior, $P(|m| + |\sigma| \geq 0.1) \approx 0.1$. Level curves are shown in black.}     
        \label{fig:Gaussian_2d_prior}
    \end{figure}
\end{example}

It is important to note that the order of parameters in which the two steps are performed may affect the final result of the two-step prior. However, a more significant issue arises with the two-step approach when the base parameter set is not unitary. To illustrate this, consider the case where we aim to obtain a two-step prior for \((\theta_1, \theta_2) \in \boldsymbol{\Theta}\). Suppose we have the base measure \(\mu_b\) and the base parameter set \(\boldsymbol{\Theta}_b = \{(\theta_1, \theta_2): \theta_1 = \theta_{1,b}\}\). In this scenario, the base model parameters correspond to fixing \(\theta_1\) at \(\theta_{1,b}\). This means there is no value of \(\theta_2\) that can be considered a base model value.
To proceed with the two-step approach, we must first penalize the distance between \(\mu_{\theta_{1,b}, \theta_2}\) and \(\mu_b\), and then penalize the distance between \(\mu_{\theta_1, \theta_2}\) and \(\mu_{\theta_{1,b}, \theta_2}\). However, this is not feasible, since \(\{(\theta_1, \theta_2): \theta_1 = \theta_{1,b}\} \subset \boldsymbol{\Theta}_b\), which implies that \(\mu_{\theta_{1,b}, \theta_2} = \mu_b\). Consequently, the distance between \(\mu_{\theta_{1,b}, \theta_2}\) and \(\mu_b\) is zero, leaving nothing to penalize.
The same problem occurs when \(\boldsymbol{\Theta}_b = \{(\theta_1, \theta_2): \theta_1 = \theta_{1,b} \text{ or } \theta_2 = \theta_{2,b}\}\).
In particular, there is no consistent way to define the two-step approach for the scale and tail index of the GP distribution considered in Section~\ref{sec:tailGP}, whereas the multivariate WCP prior is easily obtained (see Section~\ref{subsec:extremevalue_wcp_multi}).

\begin{figure}[t]
	\centering
	\includegraphics[width = 0.9\linewidth]{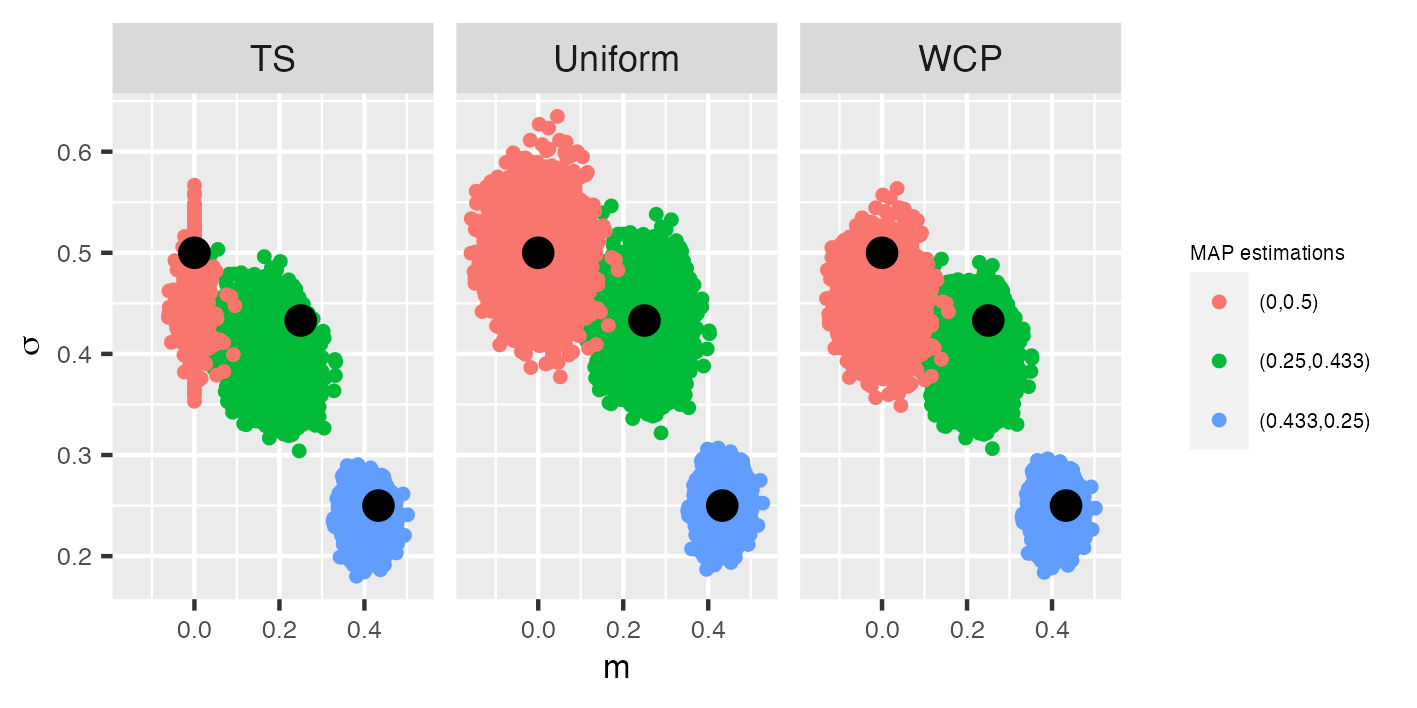}
	\caption{MAP estimations of $m$ and $\sigma$ using the two-step approach (TS), a uniform and $\text{WCP}_2$ priors. 
    Data is simulated in three scenarios using three different values of the true parameters $(m,\sigma)$. The black dots represent the true parameter values and the colored points clouds are the corresponding MAP estimates.}     
	\label{fig:sim_2dGaussian}
\end{figure}

Another drawback of the two-step approach is that it uses an approximation of the Wasserstein distance.
For example, for the Gaussian distribution parameterized by $m$ and $\sigma$, the Wasserstein distance $(m^2 + \sigma^2)^{1/2}$ is approximated by $\sigma + \abs{m}$ if $\eta_1=\eta_2$. That is, the Euclidean distance is approximated by an $L_1$ distance on $\mathbb{R}\times (0,\infty)$.
To illustrate the effect of this approximation, we compare the $\text{WCP}_2$ prior to its two-step approximation in a simulation study. 
We choose the parameters from the same level curve of the Wasserstein distance and compute their MAP estimations with 100 identically independent Gaussian data generated with the parameters. 
Figure~\ref{fig:sim_2dGaussian} shows the results based on 5000 rounds of estimations.
Compared to the uniform prior, both the $\text{WCP}_2$ prior and the two-step prior create some bias in the MAP estimates toward the base model. 
For the WCP prior, the shape of the points clouds and thus the distribution of the estimator are similar to those for the uniform prior, while for the two-step prior, they change depending on the true parameter values.
Thus, the two-step prior does not penalize equally for the same Wasserstein distance.

In Figure~\ref{fig:sim_2dGaussian}, the   hyperparameters for the two-step prior was chosen equal to that for the WCP prior.
However, one could choose the hyperparameters separately, as discussed in \citet[Section 2.2]{jasapcprior}. 
This gives the two-step approach more freedom to penalize the parameters in different ways.

\begin{remark}\label{rem:n_step_approach}
    The two-step approach can be extended to a step-wise approach for $n$ parameters as follows. Let $\boldsymbol{\theta} = (\theta_1, \ldots, \theta_n)$ be the parameter vector, which can be reordered in any convenient manner, $(\theta_{(1)}, \ldots, \theta_{(n)})$, with corresponding base model values $(\theta_{b,(1)}, \ldots, \theta_{b,(n)})$. The step-wise approach is constructed iteratively. First, compute $\pi(\theta_{(1)})$ as the WCP prior of $\theta_{(1)}$ given that $\theta_{(i)} = \theta_{b,(i)}$ for $i = 2, \ldots, n$. Next, compute $\pi(\theta_{(2)} | \theta_{(1)})$ as the WCP prior of $\theta_{(2)}$ conditioned on $\theta_{(1)}$, and $\theta_{(i)} = \theta_{b,(i)}$ for $i = 3, \ldots, n$. This process is repeated until $\pi(\theta_{(n)} | \theta_{(1):(n-1)} = \theta_{b,(1):(n-1)})$ is computed, which is the WCP prior of $\theta_{(n)}$ given $\theta_{(i)}$ for $i = 1, \ldots, n-1$. The resulting density is
    $    \pi(\theta_1, \ldots, \theta_n) = \pi(\theta_{(1)}) \prod_{i=2}^{n} \pi(\theta_{(i)} | \theta_{(1):(i-1)}).
    $
\end{remark}

\section{Applications of multivariate WCP priors}
\label{sec:app_multi_wcp}

\subsection{Bivariate prior for extreme value statistics}\label{subsec:extremevalue_wcp_multi}

In Section \ref{sec:tailGP}, we considered a $\text{WCP}_1$ prior for the tail index of the GP distribution.
We now derive the two-dimensional $\text{WCP}_1$ prior for $\xi$ and $\sigma$ of the GP distribution.
Since $\sigma$ is a scale parameter, we choose the base measure as a Dirac measure concentrated at $0$, which corresponds to $\sigma=0$. In this case  
$\boldsymbol{\Theta}_b = \{(\xi,\sigma): \sigma = 0\}$, which is an example in which the base parameter set is not unitary.

\begin{proposition}
	\label{2d_GP_WCP}
	The density of the ${WCP}_1$ prior for the parameters $(\xi,\sigma)$ of a GP distribution is
	\begin{equation}
		\label{analytic_2d_GP}
		\pi_{\sigma,\xi}(\sigma,\xi) 
		= \frac{\eta}{1-\xi}\exp\left(-\eta\frac{\sigma}{1-\xi}\right).
	\end{equation}
\end{proposition}

Figure~\ref{fig:PCW_GP_2D_density} shows two examples of the WCP$_1$ densities of $(\sigma,\xi)$ for the generalized Pareto distribution, when $\eta = 1$ and when $\eta=20$.
The prior with $\eta = 1$ concentrates near $\xi = 1$ and $\sigma = 0$.
This may seem counter-intuitive since the base model for one-dimensional $\text{WCP}_1$ prior of $\xi$ is $\xi = 0$. However, recall that the base parameter set it $\boldsymbol{\Theta}_b = \{(\xi,\sigma): \sigma = 0\}$, and note that the level curves of the Wasserstein distance are straight lines from $(\xi,\sigma) = (1,0)$ to a point on the $x$-axis. This means that the level curves are closer together for parameters close to $(\xi,\sigma) = (1,0)$. 


\begin{figure}[t]
    \centering  
    \includegraphics[width = 0.9\linewidth]{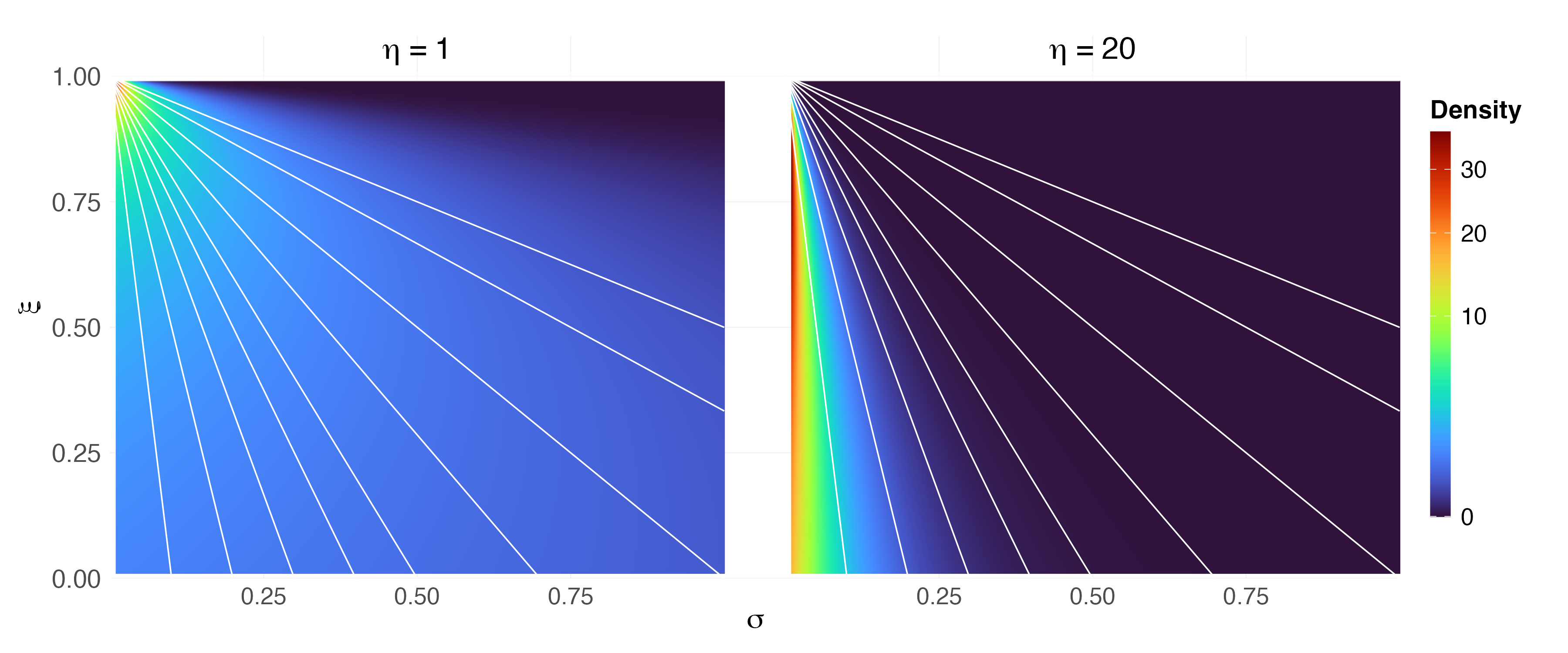}  
    \caption{$\text{WCP}_1$ prior density of $\xi$ and $\sigma$ with $\eta = 1$ (left) and $\eta = 20$ (right), along with different level curves in white. 
    }    
    \label{fig:PCW_GP_2D_density}   
\end{figure}

\subsection{WCP priors for linear regression}
\label{sec:wcp_linear_regression}
Consider the linear regression model $\boldsymbol{Y} = \boldsymbol{X} \boldsymbol{\beta} + \boldsymbol{\varepsilon}$, where $\boldsymbol{X}$ is a given $N\times n$ design matrix, $\boldsymbol{\beta} = [\beta_1, \ldots, \beta_n]^\top \in \mathbb{R}^n$ is the vector of regression coefficients, and $\boldsymbol{\varepsilon}$ is a vector of independent, centered Gaussian variables with variance $\sigma^2$, $\boldsymbol{\varepsilon} \sim \mathcal{N}(\boldsymbol{0}, \sigma^2 \boldsymbol{I}_N)$, where $N,n\in \mathbb{N}$ and $\boldsymbol{I}_N$ is the $N\times N$ identity matrix.

We begin by demonstrating a strong connection between the Bayesian lasso prior, introduced by \citet{park2008bayesian} and \citet{hans2009bayesian}, and the step-wise $\text{WCP}_2$ prior. In particular, that the Bayesian lasso can be interpreted from the perspective of the $\text{WCP}_2$ prior, providing new insights into its properties.
This setting is similar to that in Proposition~\ref{example2} but in a high dimension. Observe that the flexible model is $\mu_{\boldsymbol{\beta}} = \mathcal{N}(\boldsymbol{X} \boldsymbol{\beta}, \sigma^2 \boldsymbol{I}_N)$. To obtain the Bayesian lasso, let the base measure be $\mu_b = \mathcal{N}(\boldsymbol{0}, \sigma^2 \boldsymbol{I}_N)$, in which case we have the base parameter set $\boldsymbol{\Theta}_b = \{\boldsymbol{0} \}$.
The step-wise WCP$_2$ prior for $\boldsymbol{\beta}$, has density
$$\pi_{\boldsymbol{\beta}}(\boldsymbol{\beta}) = \prod_{i = 1}^n \frac{\eta_i}{2} \| \boldsymbol{X}_{(i)} \|_{\mathbb{R}^N} \exp \left( -\eta_i \cdot \norm{ \boldsymbol{X}_{(i)} }_{\mathbb{R}^N} |\beta_i| \right),$$
where $\boldsymbol{X}_{(i)}$ denotes the $i$th column of $\boldsymbol{X}$, and $\eta_i>0$, $i=1,\ldots,n$ are user-specified hyperparameters.
A detailed derivation can be found in Appendix~\ref{app:wcp4regression}.
Observe that the Bayesian lasso prior coincides with the step-wise $\text{WCP}_2$ prior if we set $\eta_i \norm{\boldsymbol{X}_{(i)}}_{\mathbb{R}^N} = \frac{\lambda}{\sqrt{\sigma^2}}$, $i=1,\ldots,n$. In particular, this means that, keeping $\sigma>0$ fixed, the Bayesian lasso prior is equivalent to the step-wise $\text{WCP}_2$ prior with $\eta_i = \frac{\lambda}{\sqrt{\sigma^2} \norm{\boldsymbol{X}_{(i)}}_{\mathbb{R}^N} }$, $i=1,\ldots,n$. Furthermore, this implies that the Bayesian lasso prior contracts towards the base model $\mathcal{N}(\boldsymbol{0}, \sigma^2 \boldsymbol{I}_N)$ with respect to the Wasserstein$-2$ distance. Further observe that even for the base model itself, the probability of having sparse coefficients is zero, that is, if $\boldsymbol{\beta}\sim \mathcal{N}(\boldsymbol{0}, \sigma^2 \boldsymbol{I}_N)$, then $P(\exists j \in \{1,\ldots, n\}: \beta_j = 0) = 0$. Therefore, this is not a suitable prior if the goal is to achieve sparsity. This is in consonance with the results in \citet{castillo2015bayesian}.

To obtain the WCP$_2$ prior according to Definition~\ref{def:informal_multi}, observe that 
$$W_2(\boldsymbol{\beta}) = W_2(\mathcal{N}(\boldsymbol{X} \boldsymbol{\beta}, \sigma^2 \boldsymbol{I}_N), \mathcal{N}( \boldsymbol{0}, \sigma^2 \boldsymbol{I}_N)) = \norm{\boldsymbol{X}\boldsymbol{\beta}}_{\mathbb{R}^n}.$$ Therefore, level sets of $W_2(\boldsymbol{\beta})$ are $(n-1)$-dimensional ellipsoids in $\mathbb{R}^{n}$. Explicit expression for the surface area of such level sets are thus available \citep{rivin2007surface}. By using the spherical parameterization of the ellipsoids, the prior density is obtained as 
\begin{equation}\label{eq:WCP2_beta}
    \pi_{\boldsymbol{\beta}}(\boldsymbol{\beta}) = \frac{\eta\exp\left(-\eta \|\boldsymbol{X}\boldsymbol{\beta}\|_{\mathbb{R}^n}\right)}{\hbox{Area}_{n-1}(S_{W_2(\boldsymbol{\beta}), \boldsymbol{\beta}})\sqrt{D(\|\boldsymbol{X}\boldsymbol{\beta}\|_{\mathbb{R}^n}, \boldsymbol{\beta})}}.
\end{equation}
by a direct application of Definition~\ref{multi_pcwprior_defi}. The derivation and the expressions of the area and the function $D(\cdot, \cdot)$ can be found in Appendix~\ref{app:wcp4regression}. 

\begin{figure}[t]
    \centering  
    \includegraphics[width = 1\linewidth]{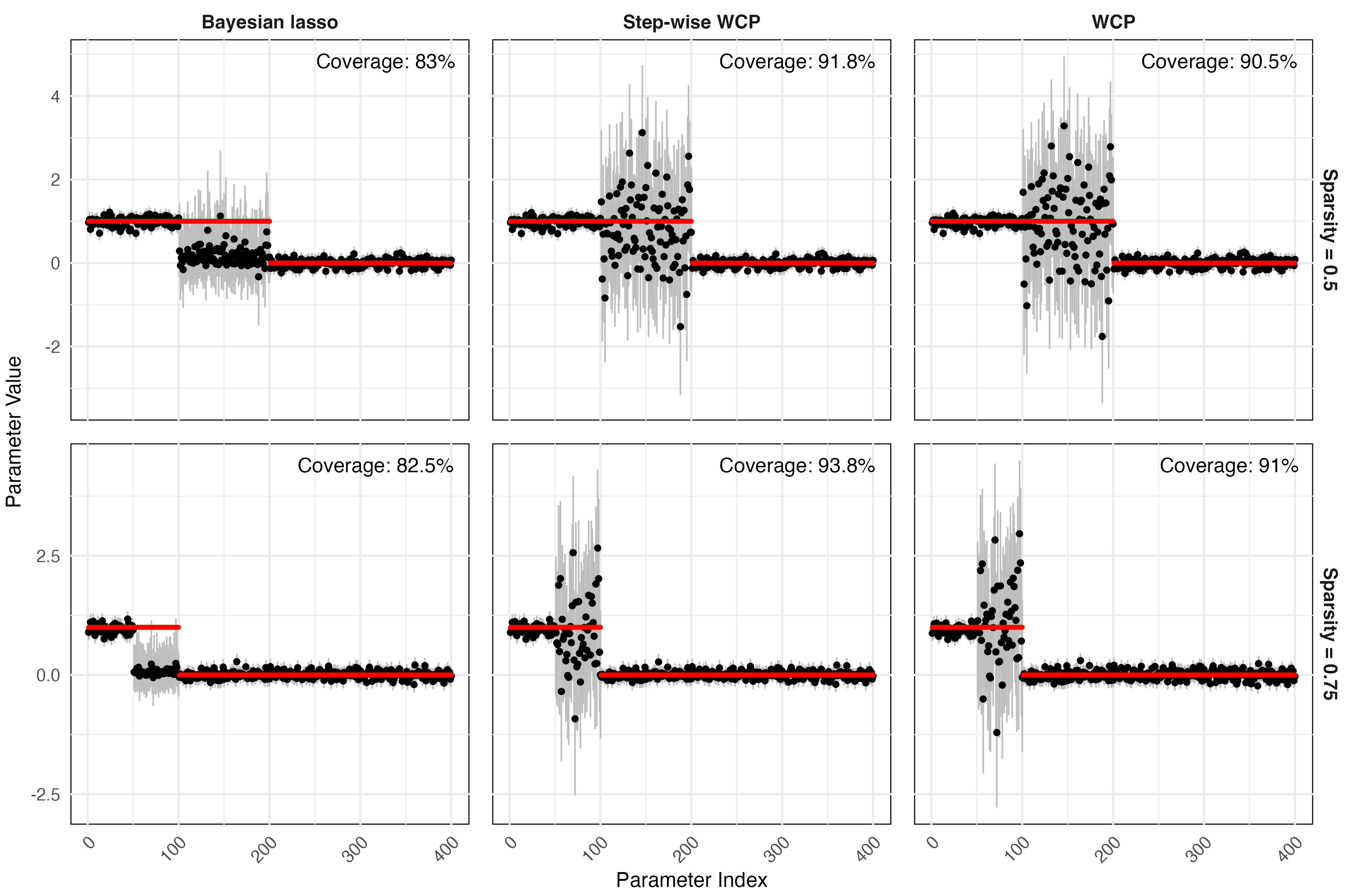}
    \caption{90\% credible bands are shown based on the posterior samples for $\boldsymbol{\beta}$ (400 parameters) under different sparsity levels with the Bayesian lasso prior, step-wise WCP$_2$ prior and WCP$_2$ prior. The X-axis shows the index $i$ of $\beta_i$ , posterior means are show as black dots and the true values of are represented as red line segments. }    
    \label{fig:comparison_200}   
\end{figure}
To compare the two WCP$_2$ priors with the Bayesian lasso prior, we consider a simulation study with two scenarios with overfitted models under different sparsity levels. In both scenarios we have $n=N=400$ and \(\boldsymbol{X}\) is a diagonal matrix with entries \(X_{i,i} = 1\) for \(i = 1,\ldots, n_1\), \(X_{i,i} = 0.1\) for \(i = n_1+1,\ldots, n_2\), \(X_{i,i} = 1\) for \(i > n_2\) and we let \(\beta_i = 1\) for \(i = 1,\ldots, n_1+n_2\) and \(\beta_i = 0\) for \(i > n_1+n_2\) and $\sigma = 0.1$. In the first scenario we have $n_1 = n_2 = 50$ and in the second $n_1 = n_2 = 100$.
The models are fitted via MCMC by \texttt{Stan} \citep{rstan}, $\sigma$ is kept fixed at $0.1$, and $\eta$ is estimated assuming an improper uniform prior on $\mathbb{R}^{+}$.

Figure~\ref{fig:comparison_200} shows that the Bayesian lasso results a highly biased posterior for $\beta_i$ with covariates being 0.1.
This effect is related to the sparsity of $\boldsymbol{\beta}$. 
The WCP$_2$ priors demonstrate superior performance compared to the Bayesian lasso prior in both cases, primarily due to their improved coverage of credible bands. This enhanced coverage ensures more reliable uncertainty quantification, making the WCP$_2$ priors better suited for capturing the true parameter values within the credible intervals. The reason is that the Bayesian lasso prior ignores the scales of the covariates, over-penalizing the coefficients of the covariates with smaller scales.

\section{Discussion}
\label{discussion}
We introduced WCP priors as a principled alternative to PC priors, replacing the Kullback--Leibler divergence with the Wasserstein distance and generalizing the concept of a ``base model'' to that of a base measure. These modifications ensure that WCP priors adhere to their principles without requiring alternative principles, such as Principle 5 in the PC prior framework. By using the Wasserstein distance, WCP priors avoid issues with infinite divergences, provide a genuine metric interpretation of complexity, and retain invariance under smooth reparameterizations.

We demonstrated that WCP priors are mathematically tractable and practical to use. Appendix~\ref{app:wasserstein} provides numerous examples showing that the Wasserstein distance can often be computed analytically or via numerical integration, even in cases with multiple parameters or when the base model does not induce a Dirac measure. Thus, the evaluation of the Wasserstein distance is typically not an obstacle. Additionally, we provided recipes for computing multivariate WCP priors analytically. These methods are implemented in the \texttt{R} package \texttt{WCPprior}, which includes interfaces for \texttt{stan} and \texttt{R-INLA}, facilitating practical applications. The package also contains implementations of numerical methods for approximating WCP priors when analytical solutions are infeasible.

We considered a range of applications to illustrate the flexibility and utility of WCP priors. These include location-scale families, AR(1) processes, tail index of generalized Pareto, joint Gaussian mean and standard deviation priors, and a reinterpretation of the Bayesian lasso from the step-wise WCP$_2$ viewpoint. These examples demonstrate that WCP priors are applicable in all cases where PC priors have been used, while also providing new insights and interpretations.

A promising direction for future work, particularly relevant to spatial statistics, is the development of WCP priors for random fields, such as Gaussian fields with Matérn covariance functions on bounded subsets of $\mathbb{R}^d$.

\appendix
\section{Details about the Wasserstein distance}\label{app:wasserstein}
    For a fixed $x_0 \in \mathcal{X}$, $P_p(\mathcal{X})$ is defined as
\begin{equation}
	P_p(\mathcal{X}):= \Bigl\{ \mu \in P(\mathcal{X});\quad \int_{\mathcal{X}} d(x_0,x)^p d\mu < +\infty \Bigr\},
	\nonumber
\end{equation}
where $P(\mathcal{X})$ denotes a space of probability measure on $\mathcal{X}$. 
By the triangle inequality, with a fixed $x_0$, if ${\int_{\mathcal{X}} d(x_0,x)^p d\mu < +\infty}$, then $\int_{\mathcal{X}} d(x_1,x)^p d\mu < +\infty$ for any other $x_1 \in \mathcal{X}$. 
Therefore, the definition of $P_p(\mathcal{X})$ does not depend on the choice of $x_0$.
    If $\mu,\nu\in P_p(\mathcal{X})$, one can show that $W_p(\mu,\nu)<\infty$, and $W_p$ is in fact a distance on $P_p(\mathcal{X})$ \citep[][p.94]{optransoldnew} so that $(P_p(\mathcal{X}),W_p)$ is a metric space.

    As expected from its definition, the Wasserstein distance is generally challenging to compute. However, for measures on $\mathbb{R}$, the following result, proven by \citet{irpino2015basic}, provides a useful simplification.    
    
    \begin{remark}\label{remarkwformula}
        Let $\mu$ and $\nu$ be two probability measures  on $\mathbbm{R}$. Then 
        $$
        W_p^p(\mu,\nu) =  \int_0^1 |F_{\mu}^{-1}(t)-F_{\nu}^{-1}(t)|^p dt \quad \text{and} \quad W_1(\mu,\nu) = \int_{\mathbbm{R}}|F_{\mu}(x) - F_{\nu}(x)|dx,
        $$ 
        where $F_{\mu}$ and $F_{\nu}$ are the distribution functions of $\mu$ and $\nu$, with corresponding pseudo-inverse  $F_i^{-1}(t) = \inf\{x:F_i(x)\geq t\}$ for $i = \mu,\nu$.
        If $m_{\mu}$ and $m_{\nu}$ are the means, $\sigma^2_{\mu}$ and $\sigma^2_{\nu}$ the variances, and 
        $$
        \rho_{\mu,\nu} = \sigma_{\mu}^{-1} \sigma_{\nu}^{-1} \int_0^1 (F_{\mu}^{-1}(t) - m_{\mu})(F_{\nu}^{-1}(t) - m_{\nu}) dt
        $$ 
        is the correlation between $\mu$ and $\nu$, then
        $$
        W_2^2(\mu, \nu) = (m_{\mu} - m_{\nu})^2 + (\sigma_{\mu} - \sigma_{\nu})^2 + 2\sigma_{\mu}\sigma_{\nu}(1 - \rho_{\mu,\nu}).
        $$
    \end{remark}
    
    The formulas in Remark~\ref{remarkwformula} can be used to compute $W_p(\mu,\nu)$ analytically in many cases, and if the integrals cannot be evaluated analytically, they can be approximated numerically as we will explore later. It should also be noted that the measures $\mu$ and $\nu$ can depend on several parameters, so the formulas are not restricted to priors of distributions with only one parameter.
    
    The following proposition shows that the Wassertein-$p$ distance is much simpler to compute when one of the measures is a Dirac measure. This is important as those are the only required distances for WCP priors if the base model is a Dirac measure. In the next proposition and in the remainder of the paper, $\norm{\cdot}_{\mathbb{R}^d}$ is the Euclidean norm on $\mathbb{R}^d$.
    
    \begin{proposition}
        \label{loc-scale-prop}
        Let $\boldsymbol{X}:\Omega \rightarrow \mathbb{R}^d$ be a random variable on a probability space $(\Omega, \mathcal{F}, P)$ with finite $p$th moment, $p\geq 1$. Let $\mu_{\boldsymbol{X}}$ be its distribution on $\mathbb{R}^d$, and let $\delta_{\boldsymbol{s}}$ be the Dirac measure supported on $\boldsymbol{s}\in\mathbb{R}^d$.
        Then, for $a \in\mathbb{R}$ and $\boldsymbol{b} \in \mathbb{R}^d$,
        ${W_p(\delta_{\boldsymbol{s}},\mu_{a\boldsymbol{X}+\boldsymbol{b}}) = \left(\mathbb{E} \norm{a\boldsymbol{X} + \boldsymbol{b} - \boldsymbol{s}}_{\mathbb{R}^d}^p \right)^{1/p}}$.
        In particular, if $\boldsymbol{b} = \boldsymbol{s}$, then $W_p(\delta_{\boldsymbol{s}},\mu_{a\boldsymbol{X}+\boldsymbol{s}}) = a\left(\mathbb{E} \norm{\boldsymbol{X}}_{\mathbb{R}^d}^p \right)^{1/p}$.
    \end{proposition}

    \begin{proof}
        The only coupling between $\mu_{aX+b}$ and $\delta_{s}$ is the independent coupling. That is, the only element 
        in $\Pi(\mu_{aX+b}, \delta_{s})$ from  Definition~\ref{wassersteindis_defi} is $\pi(A\times B) = \mu_{aX+b}(A)\delta_s(B)$
        for any $A,B\in \mathcal{B}(\mathbb{R}^d)$, where $\mathcal{B}(\mathbb{R}^d)$ denotes the Borel $\sigma$-algebra on $\mathbb{R}^d$. 
        Therefore, the conclusion follows directly. 
    \end{proof}

    Thus, when the base model with corresponding probability distribution $\mu_b$ is a Dirac measure, $W_p(\mu_b, \mu_{\theta})$ can be obtained as a $p$th moment, which often is available in closed form.
    Another important special case is when $\mu$ and $\nu$ are two Gaussian measures on $\mathbbm{R}^d$ with mean vectors $\boldsymbol{m}_\mu, \boldsymbol{m}_\nu$ and non-singular covariance matrices $\boldsymbol{\Sigma}_\mu, \boldsymbol{\Sigma}_\nu$, respectively. 
    Then, by \citet[][Proposition 7]{WDfiniteD},
    \begin{equation}
        W_2^2(\mu,\nu) = 
        \norm{\boldsymbol{m}_\mu-\boldsymbol{m}_\nu}_{\mathbbm{R}^d} + \text{tr}(\boldsymbol{\Sigma}_\mu) + \text{tr}(\boldsymbol{\Sigma}_\nu) - 2\text{tr}( (\boldsymbol{\Sigma}_\mu^{1/2}\boldsymbol{\Sigma}_{\nu}\boldsymbol{\Sigma}_{\mu}^{1/2})^{1/2} ),
        \label{w2dist2gaussianfd}
    \end{equation}
    where $\text{tr}(\cdot)$ is the trace of a matrix.

\section{Base models as base measures}\label{app:basemeasure}
In this section, we provide a detailed description of the assumptions regarding the base model measure. Let $\boldsymbol{\Theta} \subset \mathbb{R}^d$, with $q \in \mathbb{N}$, be a parameter set. Fix some $p \geq 1$ and assume that the model set $M = \{\mu_{\boldsymbol{\theta}}: \boldsymbol{\theta} \in \boldsymbol{\Theta}\}$ consists of probability measures defined on a metric space $(\mathcal{X}, d)$, and that $M \subset P_p(\mathcal{X})$. 

The extended model set $\overline{M}$ is defined as the closure of $M$ in $P_p(\mathcal{X})$, meaning that $\mu \in \overline{M}$ if and only if there exists a sequence $(\boldsymbol{\theta}_n) \subset \boldsymbol{\Theta}$ such that $W_p(\mu_{\boldsymbol{\theta}_n}, \mu) \to 0$ as $n \to \infty$. 

\begin{remark}
A sequence of probability measures $\{\mu_k\}_k\in P_p(\mathcal{X})$ converges to 
$\mu \in P_p(\mathcal{X})$ with distance $W_p$ if and only if $\{\mu_k\}_k$ converges weakly to $\mu$ in $P_p(\mathcal{X})$. In particular,
convergence with respect to the Wasserstein distance $W_p$ is equivalent to convergence of the $p$th moment and weak convergence in $P(\mathcal{X})$ (i.e., the usual weak convergence of measures, that induces convergence in distribution of random variables) \citet[Definition~6.8]{optransoldnew}. 
\end{remark}

We now examine the relationship between $\overline{M}$ and $\boldsymbol{\Theta}$. Suppose that $\boldsymbol{\Theta} = \Theta_1 \times \Theta_2 \times \cdots \times \Theta_d$, where $\Theta_i \subset \mathbb{R}$ for $i = 1, \ldots, d$. Let $\mu_b \in \overline{M}$ be the base measure, and let $(\boldsymbol{\theta}_n)$ be a sequence such that $W_p(\mu_{\boldsymbol{\theta}_n}, \mu_b) \to 0$ as $n \to \infty$. We can write $\boldsymbol{\theta}_n = (\theta_{n,1}, \ldots, \theta_{n,d})$, where each $(\theta_{n,i})_n$ is a sequence in $\mathbb{R}$, $i=1, \ldots, d$. 
We select a monotonic subsequence $(\theta_{n_k^1,1})$ such that $\theta_1$ exists in the extended real line $\overline{\mathbb{R}}:=[-\infty, \infty]$ with $\theta_{n_k^1,1} \to \theta_1$. Proceeding inductively, let $(n_k^2)$ be a subsequence of $(n_k^1)$ such that $(\theta_{n_k^2,2})$ is monotonic, ensuring the existence of $\theta_2$ in the extended real line such that $\theta_{n_k^2,2} \to \theta_2$. Continuing this process, we obtain a sequence $(n_k^d)$ such that there exists $\boldsymbol{\theta}_b = (\theta_{b,1}, \ldots, \theta_{b,d}) \in [-\infty, \infty]^d$, and for each $i = 1, \ldots, d$, we have $\theta_{n_k^d,i} \to \theta_{b,i}$. Therefore, $\boldsymbol{\theta}_b \in \overline{\boldsymbol{\Theta}} := \overline{\Theta}_1 \times \cdots \times \overline{\Theta}_d$, where the closures $\overline{\Theta}_i$, for $i = 1, \ldots, d$, are taken in the extended real line. Thus, we can define $\mu_{\boldsymbol{\theta}_b} := \mu_b$. 

Now, we define the base parameter set as 
$$\boldsymbol{\Theta}_b := \{\boldsymbol{\theta}\in \overline{\boldsymbol{\Theta}}: \exists (\boldsymbol{\theta}_n)\subset \boldsymbol{\Theta}, \boldsymbol{\theta}_n \to \boldsymbol{\theta}\hbox{ and } W_p(\mu_{\boldsymbol{\theta}_n},\mu_b)\to 0\}.$$

Finally, observe that if the map $\boldsymbol{\theta} \mapsto W_p(\mu_{\boldsymbol{\theta}}, \mu_b)$ is uniformly continuous, then this definition of $\mu_{\boldsymbol{\theta}_b} := \mu_b$ is unambiguous. Indeed, for convenience, let $W_p:\boldsymbol{\Theta}\to\mathbb{R}$ represent the map ${W_p(\boldsymbol{\theta}) = W_p(\mu_{\boldsymbol{\theta}}, \mu_b)}$, and observe that if $W_p$ is uniformly continuous, there exists a unique extension $\widetilde{W}_p:\overline{\boldsymbol{\Theta}} \to \overline{\mathbb{R}}$, thus for any sequence $\boldsymbol{\theta}_n\to\boldsymbol{\theta}_b$, we have
$$W_p(\boldsymbol{\theta}_n) = \widetilde{W}_p(\boldsymbol{\theta}_n) \to \widetilde{W}_p(\boldsymbol{\theta}_b) = 0.$$

\begin{remark}
    The assumption \(\boldsymbol{\Theta} = {\Theta}_1 \times \cdots \times {\Theta}_d\) is not necessary. We assume this to prove the result in general; however, there may exist examples where \(\boldsymbol{\Theta} \subset \mathbb{R}^d\) does not take this form, yet it is still possible to find a sequence \((\boldsymbol{\theta}_n) \subset \boldsymbol{\Theta}\) and a corresponding \(\boldsymbol{\theta}_b \in \overline{\boldsymbol{\Theta}}\) such that \(W_p(\mu_{\boldsymbol{\theta}_n}, \mu_b) \to 0\) as \(n \to \infty\).
\end{remark}


\section{Technical details on WCP priors}\label{app:assumptions}
    For the univariate WCP priors, we make the following assumptions.
    \begin{assumption}\label{assump1}
        The family $(\mu_\theta)_{\theta \in \Theta}$ satisfies:
        \begin{enumerate}	
            \item Both $W_p^{-}$ and $W_p^{+}$ are injective and differentiable on their 
            domains. \label{assump1:item1}
            \item If $\Theta_-\neq\emptyset$, then we require ${\lim_{\theta \rightarrow \theta_b}W_p^{-}(\theta) = 0}$ and ${\lim_{\theta \rightarrow \theta_-}W_p^{-}(\theta) = c_-}$. 
            If $\Theta_+\neq\emptyset$, the we require ${\lim_{\theta \rightarrow \theta_b}W_p^{+}(\theta) = 0}$ and ${\lim_{\theta \rightarrow \theta_+}W_p^{+}(\theta) = c_+}$. \label{assump1:item2}
        \end{enumerate}
    \end{assumption} 

Assumption \ref{assump1}:\ref{assump1:item1} is a mild condition that allows us to perform the change of variables, whereas Assumption \ref{assump1}:\ref{assump1:item2} is a condition that allows us to ensure that the WCP prior will contract towards the base model.

In the multivariate case, we need the following weak regularity conditions on the model to guarantee that the prior is well-defined:
\begin{assumption}   
	\label{multi_pwc_assumptions}
	The family \((\mu_{\boldsymbol{\theta}})_{\boldsymbol{\theta} \in \boldsymbol{\Theta}}\) satisfies:   
	\begin{enumerate} 
		\item \(\boldsymbol{\theta} \mapsto W_p(\boldsymbol{\theta})\) is of class \(C^1\) on \(\boldsymbol{\Theta}\) and continuous on \(\overline{\boldsymbol{\Theta}}\). \label{multi_pwc_assumption:item2} 
		\item \(W_p(\boldsymbol{\theta})\) has a nonvanishing gradient for \(\boldsymbol{\theta} \neq \boldsymbol{\theta}_b\).\label{multi_pwc_assumption:item3} 
		\item For any \(w \in \mathbb{R}\) such that the level set \(S_{w,\boldsymbol{\theta}} = \{\boldsymbol{\theta}\in \overline{\boldsymbol{\Theta}} \mid W_p(\boldsymbol{\theta}) = w\}\) is nonempty, \(S_{w,\boldsymbol{\theta}}\) is a compact set.\label{multi_pwc_assumption:item4} 
	\end{enumerate} 
\end{assumption} 

Assumptions \ref{multi_pwc_assumptions}:\eqref{multi_pwc_assumption:item2} and \ref{multi_pwc_assumptions}:\eqref{multi_pwc_assumption:item3} imply that \(W_p(\cdot)\) is a submersion, which is both realistic and natural, as it suggests the absence of critical points for the Wasserstein distance among the flexible models—an expectation for any meaningful parameterization.  Additionally, Assumption~\ref{multi_pwc_assumptions}:\eqref{multi_pwc_assumption:item4} plays a crucial role in guaranteeing the existence of a uniform distribution over the level sets, as there are no uniform distributions over unbounded sets. This property will be essential for the following construction of  multivariate WCP priors.
\begin{remark}
	Observe that \(W_p(\mu_{\boldsymbol{\theta}_b}, \mu_{\boldsymbol{\theta}}) \rightarrow 0\) as \(\|\boldsymbol{\theta}_b - \boldsymbol{\theta}\|_{\mathbb{R}^d} \rightarrow 0\).
	Further, for a constant \(w > 0\), let \(S_{w,\boldsymbol{\theta}} = \{\boldsymbol{\theta} \in \overline{\boldsymbol{\Theta}} \mid W_p(\boldsymbol{\theta}) = w\}\) be the level set of \(W_p(\boldsymbol{\theta})\) corresponding to \(w\).  Assumption~\ref{multi_pwc_assumptions}:\eqref{multi_pwc_assumption:item2}-\eqref{multi_pwc_assumption:item3} and the inverse function theorem imply that \(S_{w,\boldsymbol{\theta}}\) is a \(C^1\)-hypersurface in \(\mathbb{R}^d\). 
	This also implies that 
	\(\boldsymbol{\Theta}\) is a foliation formed by the level sets \(S_{w,\boldsymbol{\theta}}\). 
\end{remark}

We will now proceed to give a general definition of the multivariate WCP priors, but first we need to introduce uniform distributions over the level sets. To this end, we will provide a brief introduction to integration on hypersurfaces, which will a key concept for such a  definition. Let $M\subset \mathbb{R}^d$ be a compact differentiable hypersurface given by a level set\footnote{In fact, every compact smooth hypersurface is a level set, which is a result that is typically used to prove the celebrated Jordan-Brouwer separation theorem.}. Further, given parameterization $X: U\subset\mathbb{R}^{d-1}\to O_M\subset M$, where $O_M$ is an open set (in the induced topology on $M$ from $\mathbb{R}^d$) and $X(\boldsymbol{u}) = (X_1(\boldsymbol{u}), \ldots, X_d(\boldsymbol{u})),$ for $\boldsymbol{u}\in U$. The area element\footnote{The area element is also commonly referred to as the volume element on $M$ or the volume form on $M$, but we choose the nomenclature area element to avoid possible confusion as $M$ is embedded in $\mathbb{R}^d$, and volume element might create a misunderstanding.} on $M$ around $O_M$ is given by
$$d_{d-1}V = \sqrt{\det \boldsymbol{G}} \, du_1 \wedge du_2 \wedge \cdots \wedge du_{d-1},$$
where $\boldsymbol{G} = \boldsymbol{G}(\boldsymbol{u}) = (g_{ij}(u))_{i,j=1}^{d-1}$ is the induced metric tensor on $M$ (the pullback from the Euclidean metric on $\mathbb{R}^d$) whose entries for $i,j=1,\ldots,d-1$ are
\begin{equation}\label{eq:metric_tensor}
    g_{ij}(\boldsymbol{u}) = \left\< \frac{\partial X}{\partial u_i}(\boldsymbol{u}), \frac{\partial X}{\partial u_j}(\boldsymbol{u})\right\> = \sum_{k=1}^d \frac{\partial X_k}{\partial u_i}(\boldsymbol{u}) \frac{\partial X_k}{\partial u_j}(\boldsymbol{u}).
\end{equation}
Now, observe that since $M$ is compact, it can be covered by a finite collection of parameterizations. That is, there exist $N\in\mathbb{N}$ such that $X_1,\ldots, X_N$ are paremeterizations on $M$, with $X_j: U_j \subset \mathbb{R}^{d-1} \to O_{M,j} \subset M$ and ${M = \bigcup_{j=1}^M O_{M,j}}$. In order to integrate with respect to the area element $d_{d-1}V$ on $M$, we must introduce the notion of partition of unity subordinated to the open cover $\{O_{M,j} \mid j=1,\ldots,N\}$:
\begin{definition}[Partition of Unity]
    Let $\{O_{M,j}\}_{j=1}^N$ be an open cover of $M$. A partition of unity subordinated to this cover is a collection of smooth functions $\{\psi_j\}_{j=1}^N$ such that:
    \begin{enumerate}
        \item $\psi_j: M \to [0,1]$ for all $j = 1, \ldots, N$;
        \item $\text{supp}(\psi_j) \subset O_{M,j}$ for all $j = 1, \ldots, N$, where $\hbox{supp}(\cdot)$ stands for the support and is defined as $\hbox{supp}(\psi_j) := \overline{\{\boldsymbol{\theta}\mid \psi_j(\boldsymbol{\theta})\neq 0\}}$;
        \item $\sum_{j=1}^N \psi_j(x) = 1$ for all $x \in M$.
    \end{enumerate}
\end{definition}

The existence of partitions of unity subordinated to an open cover is proved, e.g., in \citet{spivak1965calculus}. Now, given a partition of unity $\{\psi_j\}_{j=1}^N$ subordinated to the cover $\{O_{M,j}\}_{j=1}^N$, we can define the integral of a function $f: M \to \mathbb{R}$ with respect to the area element $d_{d-1}V$ as
\begin{equation}\label{eq:integral_manifold}
    \int_M f(x) \, d_{d-1}V = \sum_{j=1}^N \int_{U_j} f(X_j(\boldsymbol{u})) \psi_j(X_j(\boldsymbol{u})) \sqrt{\det \boldsymbol{G}_j(\boldsymbol{u})} \, du_1 \ldots du_{d-1},
\end{equation}
where $\boldsymbol{G}_j$ is the induced metric tensor corresponding to the parameterization $X_j$ and whose elements can be computed as in \eqref{eq:metric_tensor}. Now, \eqref{eq:integral_manifold} allows us to define the surface area of $M$:
\begin{equation}\label{eq:surface_area}
    \text{Area}_{d-1}(M) = \int_M d_{d-1} V,
\end{equation}
and to define the uniform distribution over $M$, which is given by the measure
$$P_M(S) = \frac{1}{\hbox{Area}_{d-1}(M)} \int_{S} d_{d-1} V.$$
We are now in a position to define the general multivariate WCP prior, that is, WCP prior for the case in which the level sets needs to be covered by more than one parameterization. Thus, assume that for every $w>0$ such that $S_{w,\boldsymbol{\theta}}\neq \emptyset$, we have that $S_{w,\boldsymbol{\theta}}$ is a compact and differentiable hypersurface. Then, the multivariate WCP prior is the distribution that assigns a (possibly truncated) exponential distribution for $w$ and, given $w$, a uniform distribution for $\boldsymbol{\theta}$ on $S_{w,\boldsymbol{\theta}}$. Therefore, the WCP prior is the probability distribution on $M_{WCP} := \bigcup_{w > 0} \{w\} \times S_{w,\boldsymbol{\theta}}$ given by
$$P( (w, \boldsymbol{\theta}) \in A) = \int_{D(A)} \frac{\eta \exp(-\eta w)}{1-\exp(-\eta c)} \left(\frac{1}{\hbox{Area}_{d-1}(S_{w,\boldsymbol{\theta}})}\int_{A_w} d_{d-1}V_w\right)\, dw,$$
where $c := \sup_{\boldsymbol{\theta}} W_p(\boldsymbol{\theta})$, which can be infinite, $d_{d-1}V_w$ is the area element in $S_{w,\boldsymbol{\theta}}$, $M_{WCP}$ is endowed the relative topology induced by $(0,\infty)\times \mathbb{R}^d$, $A$ is a Borel set in $M_{WCP} \subset (0,\infty)\times \mathbb{R}^d$, $A_w = \{\boldsymbol{\theta} \in\mathbb{R}^d \mid (w, \boldsymbol{\theta}) \in A\}$ and 
$$D(A) = \{w >0 \mid \exists \boldsymbol{\theta} \in S_{w,\boldsymbol{\theta}} \hbox{ such that } (w, \boldsymbol{\theta}) \in A\}.$$

Now, for each $w>0$, $S_{w,\boldsymbol{\theta}}$ is a compact hypersurface and, therefore, can be covered by finitely many parameterizations. Thus, let $S_{w, \boldsymbol{\theta}} = \bigcup_{k=1}^{N_w} O_{k,w}$ be an open cover induced by such parameterizations and let $\{\psi_{k,w}\}_{k=1}^{N_w}$ be an associated partition of unity subordinated to $\{O_{k,w}\}_{k=1}^{N_w}$. By letting $\psi_{k,w} \equiv 0$ and $O_{k,w} = \emptyset$ for $k>N_w$, we can write for all $w>0$, $S_{w,\boldsymbol{\theta}} = \bigcup_{k=1}^\infty O_{k,w}$ with associated partition of unity $\{\psi_{k,w}\}_{k\in\mathbb{N}}$. Further, let for each $w>0$ and each $k\in\mathbb{N}$, $X_{k,w} : U_{k,w} \subset\mathbb{R}^{d-1}\to O_{k,w}\subset S_{w,\boldsymbol{\theta}}$ be the corresponding parameterization. Since all entries are positive, we can use \eqref{eq:integral_manifold} and Fubini-Tonelli's theorem (to interchange between the infinite sum and integral and also to turn the iterated integrals into a multiple integral) to obtain that
\begin{align*}
    &\int_{D(A)} \frac{\eta \exp(-\eta w)}{1-\exp(-\eta c)} \left(\frac{1}{\hbox{Area}_{d-1}(S_{w,\boldsymbol{\theta}})}\int_{A_w} d_{d-1}V_w\right)\, dw =\\
    &\int_{A}\sum_{k=1}^\infty \mathbbm{1}_{w >0, \boldsymbol{u} \in O_{k,w}} \frac{\eta \exp(-\eta w)}{1-\exp(-\eta c)} \frac{\psi_{k,w}(\boldsymbol{u}) \sqrt{\boldsymbol{G}(\boldsymbol{u})}}{\hbox{Area}_{d-1}(S_{w,\boldsymbol{\theta}})} \, dw d\boldsymbol{u},
\end{align*}
where $\boldsymbol{u} = (u_1,\ldots,u_{d-1})$ and $\boldsymbol{G}(\boldsymbol{u})$ is given by \eqref{eq:metric_tensor}.

Finally, to obtain the multivariate WCP density, for each $k\in\mathbb{N}$, we do the change of variables induced by the map $\boldsymbol{\Phi}_k(\boldsymbol{\theta}) = (W_p(\boldsymbol{\theta}), X_{k,W_p(\boldsymbol{\theta})}^{-1}(\boldsymbol{\theta}))$ to arrive at the multivariate WCP density for $\boldsymbol{\theta}$ which is formalized in the following definition, where $J_{\boldsymbol{g}}(\boldsymbol{x})$ denotes the jacobian matrix of a differentiable function $\boldsymbol{g}$ evaluated at $\boldsymbol{x}$.

\begin{definition}\label{def:gen_WCP_density}
    Let Assumption \ref{multi_pwc_assumptions} hold. For each $w>0$ let $X_{k,w}:U_k\subset\mathbb{R}^{d-1}\to O_{k,w}\subset S_{w,\boldsymbol{\theta}}$, $k=1,\ldots, N_w$ be a collection of parameterizations of $S_{w,\boldsymbol{\theta}}$ that provides an open cover for it and $\{\psi_{k,w}\}_{k=1}^{N_w}$ be an associated partition of unity. Let $O_{k,w} = \emptyset$ and $\psi_{k,w} \equiv 0$ for $k>N_w$. Further, assume that for each $k = 1,\ldots, N_w$, the map $(w, \boldsymbol{u}) \mapsto X_{k,w}(\boldsymbol{u})$ is a local diffeomorphism and let $\boldsymbol{\Phi}_k(\boldsymbol{\theta}) = (W_p(\boldsymbol{\theta}), X_{k,W_p(\boldsymbol{\theta})}^{-1}(\boldsymbol{\theta}))$ if $O_{k,W_p(\boldsymbol{\theta})}\neq \emptyset$ and $\boldsymbol{\Phi}_k(\boldsymbol{\theta}) = 0$ if $O_{k,W_p(\boldsymbol{\theta})} = \emptyset$. The multivariate WCP density for $\boldsymbol{\theta}$ is given by
    $$
    \pi(\boldsymbol{\theta}) = \sum_{k=1}^\infty |\det J_{\boldsymbol{\Phi}_k}(\boldsymbol{\theta})| \frac{\eta\exp(-\eta W_p(\boldsymbol{\theta}))}{1-\exp(-\eta c)}\frac{\psi_{k,W(\boldsymbol{\theta})}(\boldsymbol{\theta})\sqrt{\det \boldsymbol{G}_{k}(\boldsymbol{\theta})}}{\text{Area}_{d-1}(S_{W_p(\boldsymbol{\theta}),\boldsymbol{\theta}})},
    $$
    where 
    $$
    {\boldsymbol{G}_k(\boldsymbol{\theta}) = J_{X_{k,W_p(\boldsymbol{\theta})}}(X_{k,W_p(\boldsymbol{\theta})}^{-1}(\boldsymbol{\theta}))^\top J_{X_{k,W_p(\boldsymbol{\theta})}}(X_{k,W_p(\boldsymbol{\theta})}^{-1}(\boldsymbol{\theta}))}
    $$ 
    and $J_{X_{k,W_p(\boldsymbol{\theta})}}(X_{k,W_p(\boldsymbol{\theta})}^{-1}(\boldsymbol{\theta}))$ is the $d\times (d-1)$ matrix $J_{X_{k,w}}(\boldsymbol{u})$ evaluated at $w=W_p(\boldsymbol{\theta})$ and $\boldsymbol{u} = X_{k,W_p(\boldsymbol{\theta})}^{-1}(\boldsymbol{\theta})$, $\eta > 0$ is a hyperparameter and $c := \sup_{\boldsymbol{\theta}} W_p(\boldsymbol{\theta})$, which can be infinite.
\end{definition}

\begin{remark}
    Definition~\ref{multi_pcwprior_defi} considers the case in which the level sets can be covered by a single parameterization. In this case there is no need for partitions of unity.
\end{remark}

Sometimes the Jacobian matrix of $\boldsymbol{\Phi}_k$ might not be convenient to compute. We will now provide some results that can be used to compute the multivariate WCP density without the need to compute the Jacobian matrix of $\boldsymbol{\Phi}_k$. First, we will connect $\det \boldsymbol{G}_k(\boldsymbol{\theta})$ with $\det J_{\boldsymbol{\Phi}_k}(\boldsymbol{\theta})$.

\begin{lemma}\label{lem:det_G_JPsi}
    Let the assumptions of Definition \ref{def:gen_WCP_density} hold.
    Then, the following identity holds for every $k\in\mathbb{N}$:
    \begin{align*}
        \sqrt{\det \boldsymbol{G}_k(\boldsymbol{\theta})} &= \frac{1}{|\det J_{\boldsymbol{\Phi}_k}(\boldsymbol{\theta})|\sqrt{D_k(W_2(\boldsymbol{\theta}), X_{k,W_2(\boldsymbol{\theta})}^{-1}(\boldsymbol{\theta}))}},
    \end{align*}
    where
    \begin{align}
        D_k(w, \boldsymbol{u}) &= \frac{\partial X_{k,w}(\boldsymbol{u})}{\partial w}^\top \boldsymbol{P}_{k,w}(\boldsymbol{u}) \frac{\partial X_{k,w}(\boldsymbol{u})}{\partial w}, \label{eq:expr_D}\\
    \boldsymbol{P}_{k,w}(\boldsymbol{u}) &= \boldsymbol{I}_{d-1} - J_{X_{k,w}}(\boldsymbol{u})(J_{X_{k,w}}(\boldsymbol{u})^\top J_{X_{k,w}}(\boldsymbol{u}))^{-1}\!J_{X_{k,w}}(\boldsymbol{u})^\top, \label{eq:P_matrix_D_matrix}
    \end{align}
    and $D_k(W_2(\boldsymbol{\theta}), X_{k,W_2(\boldsymbol{\theta})}^{-1}(\boldsymbol{\theta}))$ is the value of $D_k(w, \boldsymbol{u})$ when $w = W_2(\boldsymbol{\theta})$ and $\boldsymbol{u} = X_{k,W_2(\boldsymbol{\theta})}^{-1}(\boldsymbol{\theta})$.
\end{lemma}

\begin{proof}
We start by defining the map $\boldsymbol{\Psi}_k: (w,\boldsymbol{u}) \mapsto X_{k,w}(\boldsymbol{u})$. Further, as noted in Remark~\ref{remark:diffeomorphism}, $\boldsymbol{\Phi}_k$ is the inverse of $\boldsymbol{\Psi}_k$. Therefore, for every $\boldsymbol{\theta}$, $\boldsymbol{\Psi}_k(\boldsymbol{\Phi}_k(\boldsymbol{\theta})) = \boldsymbol{\theta}$. Hence, by the chain rule, we have that
$$
\boldsymbol{I}_n = \frac{\partial \boldsymbol{\Psi}_k(\boldsymbol{\Phi}_k(\boldsymbol{\theta}))}{\partial \boldsymbol{\theta}} = \frac{\partial \boldsymbol{\Psi}_k(w,\boldsymbol{\varphi})}{\partial (w,\boldsymbol{\varphi})} \cdot \frac{\partial \boldsymbol{\Phi}_k(\boldsymbol{\theta})}{\partial \boldsymbol{\theta}} = J_{\boldsymbol{\Psi}_k}(\boldsymbol{\Phi}_k(\boldsymbol{\theta})) \cdot J_{\boldsymbol{\Phi}_k}(\boldsymbol{\theta}).
$$
Thus,
\begin{equation}\label{eq:JPsi}
    J_{\boldsymbol{\Psi}_k}(\boldsymbol{\Phi}_k(\boldsymbol{\theta})) = (J_{\boldsymbol{\Phi}_k}(\boldsymbol{\theta}))^{-1}.
\end{equation}
Further, note that $\frac{\partial X_{k,w}(\boldsymbol{u})}{\partial \boldsymbol{u}} = J_{X_{k,w}}(\boldsymbol{u})$, so that the Jacobian matrix of $\boldsymbol{\Psi}_k$ is
$$
J_{\boldsymbol{\Psi}_k}(w,\boldsymbol{u}) = \begin{pmatrix}
\frac{\partial X_{k,w}(\boldsymbol{u})}{\partial w} & \frac{\partial X_{k,w}(\boldsymbol{u})}{\partial \boldsymbol{u}}
\end{pmatrix}
= 
\begin{pmatrix}
    \frac{\partial X_{k,w}(\boldsymbol{u})}{\partial w} & J_{X_{k,w}}(\boldsymbol{u})
\end{pmatrix}.
$$
Therefore, we have that 
\begin{align*}
J_{\boldsymbol{\Psi}}(w,\boldsymbol{u})^\top J_{\boldsymbol{\Psi}}(w,\boldsymbol{u}) &= \begin{pmatrix}
    \frac{\partial X_{k,w}(\boldsymbol{u})}{\partial w}^\top & J_{X_{k,w}}(\boldsymbol{u})^\top
\end{pmatrix} \begin{pmatrix}
    \frac{\partial X_{k,w}(\boldsymbol{u})}{\partial w} \\
    J_{X_{k,w}}(\boldsymbol{u})
\end{pmatrix} \\
&= \begin{pmatrix}
    \frac{\partial X_{k,w}(\boldsymbol{u})}{\partial w}^\top \frac{\partial X_{k,w}(\boldsymbol{u})}{\partial w} & \frac{\partial X_{k,w}(\boldsymbol{u})}{\partial w}^\top J_{X_{k,w}}(\boldsymbol{u}) \\
    J_{X_{k,w}}(\boldsymbol{u})^\top \frac{\partial X_{k,w}(\boldsymbol{u})}{\partial w} & J_{X_{k,w}}(\boldsymbol{u})^\top J_{X_{k,w}}(\boldsymbol{u})
\end{pmatrix}.
\end{align*}
Further, by using determinant properties for block matrices, we have that
\begin{align*}
\det J_{\boldsymbol{\Psi}}(w,\boldsymbol{u})^\top J_{\boldsymbol{\Psi}}(w,\boldsymbol{u}) &= \det(J_{X_{k,w}}(\boldsymbol{u})^\top J_{X_{k,w}}(\boldsymbol{u})) D(w,\boldsymbol{u})
\end{align*}
where $D_k(w,\boldsymbol{u})$ is given by \eqref{eq:expr_D}. Therefore,
$$
D_k(w,\boldsymbol{u}) \det J_{X_{k,w}}(\boldsymbol{u})^\top J_{X_{k,w}}(\boldsymbol{u}) = \det J_{\boldsymbol{\Psi}}(w,\boldsymbol{u}) = \left(\det J_{\boldsymbol{\Psi}}(w,\boldsymbol{u})\right)^2.$$
By evaluating the above expression at $w = W_2(\boldsymbol{\theta})$, and $\boldsymbol{\varphi} = X_{W_2(\boldsymbol{\theta})}^{-1}(\boldsymbol{\theta})$ and substituting this into \eqref{eq:JPsi}, we have that
\begin{align*}
    \sqrt{\det \boldsymbol{G}_k(\boldsymbol{\theta})} &= \sqrt{\det J_{X_{k,W_2(\boldsymbol{\theta})}}(X_{k,W_2(\boldsymbol{\theta})}^{-1}(\boldsymbol{\theta}))^\top J_{X_{k,W_2(\boldsymbol{\theta})}}(X_{k,W_2(\boldsymbol{\theta})}^{-1}(\boldsymbol{\theta}))}\\
     &= \frac{1}{|\det J_{\boldsymbol{\Phi}_k}(\boldsymbol{\theta})|\sqrt{D_k(W_2(\boldsymbol{\theta}), X_{k,W_2(\boldsymbol{\theta})}^{-1}(\boldsymbol{\theta}))}}.
\end{align*}
\end{proof}

\begin{remark}
    The matrix $\boldsymbol{P}_{k,w}(u)$ in \eqref{eq:P_matrix_D_matrix} is a projection matrix onto the normal space of the level set $S_{w,\boldsymbol{\theta}}$ at the point $X_{k,w}(u) = \boldsymbol{\theta}$. Furthermore, the quantity $D_k(w,\boldsymbol{u})$ in \eqref{eq:expr_D} can be interpreted as the squared norm of the residuals of a linear regression of $\frac{\partial X_{k,w}(\boldsymbol{u})}{\partial w}$ on $J_{X_{k,w}}(\boldsymbol{u})$, that is, a regression of $\frac{\partial X_{k,w}(\boldsymbol{u})}{\partial w}$ on the tangent space of $S_{w,\boldsymbol{\theta}}$ at the point $X_{k,w}(u) = \boldsymbol{\theta}$. This means that $D_k(w,\boldsymbol{u})$ is a measure of the ``non-tangential'' variation of $X_{k,w}(\boldsymbol{u})$ with respect to $w$.
\end{remark}

The above lemma allows us to provide the following alternative expression for the multivariate WCP prior.

\begin{proposition}\label{prp:multi_param_wcp_D_matrix}
    Let the assumptions of Definition \ref{def:gen_WCP_density} hold. Then, the multivariate WCP prior can be computed as
    $$
    \pi(\boldsymbol{\theta}) = \sum_{k=1}^\infty\frac{\eta\exp(-\eta W_p(\boldsymbol{\theta}))}{1-\exp(-\eta c)}\frac{\psi_{k,W(\boldsymbol{\theta})}(\boldsymbol{\theta})}{\text{Area}_{d-1}(S_{W_p(\boldsymbol{\theta}),\boldsymbol{\theta}}) D_k(W_2(\boldsymbol{\theta}), X_{k,W_2(\boldsymbol{\theta})}^{-1}(\boldsymbol{\theta}))},
    $$
    where $\eta > 0$ is a user-specified hyperparameter and $D_k(w,\boldsymbol{u})$ is given by \eqref{eq:expr_D}, with the evaluation of $D_k(w,\boldsymbol{u})$ at $w = W_2(\boldsymbol{\theta})$ and $\boldsymbol{u} = X_{k,W_2(\boldsymbol{\theta})}^{-1}(\boldsymbol{\theta})$ being $D_k(W_2(\boldsymbol{\theta}), X_{k,W_2(\boldsymbol{\theta})}^{-1}(\boldsymbol{\theta}))$ .
\end{proposition}

Similarly, if we have a single parameterization as in Definition~\ref{multi_pcwprior_defi}, then the multivariate WCP prior can be computed as
\begin{proposition}\label{prp:single_param_wcp_D_matrix}
    Let the assumptions of Definition~\ref{multi_pcwprior_defi} hold. Further, recall that in this case, for each $w>0$ where $S_{w,\boldsymbol{\theta}}\neq \emptyset$, we assume that a parameterization ${X_w: U_w \subset \mathbb{R}^{d-1} \to \widetilde{S}_{w,\boldsymbol{\theta}} \subset S_{w,\boldsymbol{\theta}}}$ exists such that $\text{Area}_{d-1}(S_{w,\boldsymbol{\theta}}\setminus \widetilde{S}_{w,\boldsymbol{\theta}}) = 0$. Then, the multivariate WCP prior can be computed as
    \begin{equation*}
        \label{multi_pcwprior_density_single_param_D_matrix}
        \pi(\boldsymbol{\theta}) = \frac{\eta\exp(-\eta W_p(\boldsymbol{\theta}))}{(1-\exp(-\eta c))\text{Area}_{d-1}(S_{W_p(\boldsymbol{\theta}),\boldsymbol{\theta}}) D(W_2(\boldsymbol{\theta}), X_{k,W_2(\boldsymbol{\theta})}^{-1}(\boldsymbol{\theta}))},
    \end{equation*}
    where $\eta > 0$ is a user-specified hyperparameter and 
    \begin{equation*}\label{eq:D_matrix_single_param}
        D(w,\boldsymbol{u}) = \frac{\partial X_{w}(\boldsymbol{u})}{\partial w}^\top \boldsymbol{P}_{w}(\boldsymbol{u}) \frac{\partial X_{w}(\boldsymbol{u})}{\partial w},
    \end{equation*}
    with
    $$
        \boldsymbol{P}_{w}(\boldsymbol{u}) = \boldsymbol{I}_{d-1} - J_{X_{w}}(\boldsymbol{u})(J_{X_{w}}(\boldsymbol{u})^\top J_{X_{w}}(\boldsymbol{u}))^{-1}\!J_{X_{w}}(\boldsymbol{u})^\top,
    $$
    and $D(W_2(\boldsymbol{\theta}), X_{W_2(\boldsymbol{\theta})}^{-1}(\boldsymbol{\theta}))$ is the value of $D(w,\boldsymbol{u})$ evaluated at $w = W_2(\boldsymbol{\theta})$ and $\boldsymbol{u} = X_{W_2(\boldsymbol{\theta})}^{-1}(\boldsymbol{\theta})$.
\end{proposition}

Observe that it is possible for the same family of measures to have different parameterizations where one parameterization satisfies Assumption \ref{multi_pwc_assumptions} while another does not. For instance, consider the family of univariate Gaussian distributions. Let $\boldsymbol{\theta}_1 = (m, \sigma)$ represent the standard parameterization with mean $m$ and standard deviation $\sigma$, and let $\boldsymbol{\theta}_2 = (m, \tau)$ represent an alternative parameterization with mean $m$ and precision $\tau = 1/\sigma^2$.
Further, let $\mu_b = \delta_0$, which corresponds to $\sigma=m=0$. Using \eqref{w2dist2gaussianfd}, the Wasserstein distance for the parameterization $\boldsymbol{\theta}_1$ is given by
$W_2(\boldsymbol{\theta}_1) = \sqrt{m^2 + \sigma^2}$,
while for $\boldsymbol{\theta}_2$ it is
$W_2(\boldsymbol{\theta}_2) = \sqrt{m^2 + \tau^{-1}}$.
Thus, the level sets \( S_{w,\boldsymbol{\theta}_1} \) for $\boldsymbol{\theta}_1$ are either empty or circles, whereas the level sets \( S_{w,\boldsymbol{\theta}_2} \) for $\boldsymbol{\theta}_2$ are either empty or unbounded. This discrepancy illustrates that Assumption \ref{multi_pwc_assumptions}:\ref{multi_pwc_assumption:item4} may be violated depending on the choice of parameterization. 
In view of the previous discussion, we can extend the definition of WCP$_p$ priors to parameterizations that violate Assumption \ref{multi_pwc_assumptions} in the following manner. 
\begin{definition}\label{def:multi_wcp_ref_par}
    Fix a reference parameterization $\boldsymbol{\theta}$ such that Assumption \ref{multi_pwc_assumptions} holds for such parameterization. This means that for every $w$, $S_{w,\boldsymbol{\theta}}$ is compact. For each $w>0$ let $\{X_{k,w,\boldsymbol{\theta}}\}_{k\in\mathbb{N}}$ be a family of parameterizations given as in Definition \ref{def:gen_WCP_density} with respect to the level sets $S_{w,\boldsymbol{\theta}}$. Let, now, $\boldsymbol{\vartheta}$ be any parameterization such that the transformation \(\boldsymbol{\theta} = g(\boldsymbol{\vartheta})\) is twice differentiable, invertible, and has a twice differentiable inverse. The multivariate WCP$_p$ prior for $\boldsymbol{\vartheta}$ based on the reference parameterization $\boldsymbol{\theta}$ is given by the following change of variables:
    $$\pi_{\boldsymbol{\theta}}(\boldsymbol{\vartheta}) := \pi(g(\boldsymbol{\vartheta})) |\det J_g(\boldsymbol{\vartheta})|,$$
    where $\pi(\cdot)$ is the multivariate WCP prior for $\boldsymbol{\theta}$ given in Definition \ref{def:gen_WCP_density}.
\end{definition}

\section{Recipes for multivariate priors}\label{app:recipe}

In this section we provide recipes for computing the multivariate WCP prior, when explicit expressions for $\hbox{Area}_{d-1}(S_{w,\boldsymbol{\theta}})$ are unknown. The following proposition is an immediate consequence of Definition \ref{def:gen_WCP_density} and of the definition of area-preserving parameterizations. More precisely, in the following recipe, level sets of $W_p(\boldsymbol{\theta})$ must be bounded hypersufaces. Assigning a uniform distribution on such hypersurface is done via an area-preserving parameterization which is a mapping from an Euclidean space to the hypersurface itself. Such parameterization guarantees that a mapped uniformly distributed random vector on that Euclidean space also follows a uniform distribution on the hypersurface.

\begin{proposition}[General recipe for computing multivariate WCP priors]
	\label{prp:recipe_multi_pcwprior}
    Let the conditions in Assumption~\ref{multi_pwc_assumptions} be satisfied. Additionally, for every $w > 0$ where $S_{w,\boldsymbol{\theta}} \neq \emptyset$, suppose there exists an area-preserving parameterization $\gamma_w: U_w \rightarrow S_{w,\boldsymbol{\theta}}$ of $S_{w,\boldsymbol{\theta}}$. Moreover, assume that $U_w = U_{1,w} \times \dots \times U_{d-1,w}$, where $U_{i,w} \subset \mathbb{R}$ are open intervals for $i = 1, \dots, d-1$.
    Let $\{u_{i,w}\}_{i=1}^{d-1}$ represent the parameters of $S_{w,\boldsymbol{\theta}}$ under this parameterization. Then, the multivariate \(\text{WCP}_p\) prior density for \(\boldsymbol{\theta}\) is
    \begin{equation}
        \pi(\boldsymbol{\theta}) =
        \abs{\det J_{W_p, \{u_{i, W_p(\boldsymbol{\theta})}\}_{i = 1}^{d-1}}(\boldsymbol{\theta})} \frac{\eta \exp(-\eta W_p(\boldsymbol{\theta}))}{1 - \exp(-\eta c)} \prod_{i = 1}^{d-1} \frac{\mathbbm{1}_{u_{i, W_p(\boldsymbol{\theta})} \in U_{i, W_p(\boldsymbol{\theta})}}} {\lambda(U_{i, W_p(\boldsymbol{\theta})})},
        \label{eq:multi_dim_wcp}
    \end{equation}
    where \( J_{W_p, \{u_{i, W_p(\boldsymbol{\theta})}\}_{i = 1}^{d-1}}(\boldsymbol{\theta}) \) denotes the Jacobian of \((W_p(\boldsymbol{\theta}), \{u_{i, W_p(\boldsymbol{\theta})}\}_{i = 1}^{d-1})\) evaluated at \(\boldsymbol{\theta}\), and \(\eta > 0\) is a user-specified hyperparameter.
\end{proposition}

The goal of introducing the area-preserving parameterization $\gamma_w$ of the level set $S_{w,\boldsymbol{\theta}}$ is to to generate a uniform distribution on it. This parameterization is a bijection from a ${d-1}$ dimensional Euclidean space $U_w$ to the ${d-1}$ dimensional hypersurface $S_{w,\boldsymbol{\theta}}$ such that for any two Borel sets $A_1, A_2 \subset U_w$, where $\lambda_{d-1}$ denotes the Lebesgue measure on $\mathbb{R}^{d-1}$, if ${\lambda_{d-1}(A_1) = \lambda_{d-1}(A_2)}$, then $\hbox{Area}_{d-1}(\gamma(A_1)) = \hbox{Area}_{d-1}(\gamma(A_2))$ and, in particular, $\lambda_{d-1}(U_w) = \hbox{Area}_{d-1}(S_{w,\boldsymbol{\theta}})$. Thus, if a random vector $u$ has a uniform distribution on $U_w$, then $\gamma_w(u)$ is uniformly distributed on $S_{w,\boldsymbol{\theta}}$.
The probability density of $\boldsymbol{u}$ is therefore $\frac{\mathbbm{1}_{\boldsymbol{u}\in U_{w}}}{\hbox{Area}_{d-1}(S_{w,\boldsymbol{\theta}})} =
\prod_{i = 1}^{d-1} \frac{\mathbbm{1}_{u_{i,w}\in U_{i,w}}} {\lambda_{d-1}(U_{i,w})}$. 
Concrete recipes for deriving the multivariate WCP priors, which also further explain the idea of Proposition~\ref{prp:recipe_multi_pcwprior}, are provided below.

\begin{remark}\label{rem:domain_prod_intervals}
    The assumption that the domain of the parameterizations of the level curves $\gamma_w$ is a Cartesian product of open intervals is not necessary and is primarily used to obtain a uniform distribution on $S_{w,\boldsymbol{\theta}}$. However, if an explicit expression for the probability density function of the uniform distribution on $S_{w,\boldsymbol{\theta}}$ is available and based on $g_{U,w}(\cdot)$ which denotes parameters ${\widetilde{\boldsymbol{u}}_w = \widetilde{\boldsymbol{u}}_w(\boldsymbol{\theta}) = (\widetilde{u}_{1,w}(\boldsymbol{\theta}),\ldots,\widetilde{u}_{d-1,w}(\boldsymbol{\theta}))}$. Then, one can directly use it to replace equation~\eqref{multi_pcwprior_density_no_change_of_variables}, yielding:
    $$
    \pi(w,\boldsymbol{u}) = \frac{\eta\exp(-\eta w)}{1 - \exp(-\eta c)} g_{U,w}(\boldsymbol{u}).
    $$
    Now, assume that the map $\boldsymbol{\theta} \mapsto \widetilde{\boldsymbol{u}}_{W_p(\boldsymbol{\theta})}(\boldsymbol{\theta})$ is differentiable. Further, let $\boldsymbol{\Phi}: \boldsymbol{\Theta} \to \mathbb{R}^{d}$ be the map  $\boldsymbol{\Phi}(\boldsymbol{\theta}) = (W_p(\boldsymbol{\theta}), \widetilde{\boldsymbol{u}}_{W_p(\boldsymbol{\theta})}(\boldsymbol{\theta}))$. Then, the multivariate WCP prior is given by
    $$
    \pi(\boldsymbol{\theta}) = \frac{\eta\exp(-\eta W_p(\boldsymbol{\theta}))}{1 - \exp(-\eta c)} g_{U,W_p(\boldsymbol{\theta})}(\widetilde{\boldsymbol{u}}_{W_p(\boldsymbol{\theta})}(\boldsymbol{\theta})) \abs{\det J_{\boldsymbol{\Phi}}(\boldsymbol{\theta})},
    $$
    where $J_{\boldsymbol{\Phi}}(\boldsymbol{\theta})$ is the jacobian matrix of $\boldsymbol{\Phi}$ evaluated at $\boldsymbol{\theta}$.
\end{remark}

We will now provide an explicit recipe for computing the quantities in Proposition \ref{prp:recipe_multi_pcwprior}. We assume that \(\boldsymbol{\Theta} = \prod_{i=1}^d \Theta_i\), where each \(\Theta_i\) is an open interval for \(i = 1, \ldots, d\). 
Let $\mu_b$ be the base measure and suppose that the base parameter set $\boldsymbol{\Theta}_b$ is connected and is contained in \(\overline{\boldsymbol{\Theta}} \setminus \boldsymbol{\Theta}\). 

\begin{recipe}
    \label{multi_pcwprior_recipe}
    Suppose we have a model with \(n\) parameters \(\boldsymbol{\theta} = (\theta_1, \ldots, \theta_n)\), and the level set \(S_{w, \boldsymbol{\theta}}\) admits a parameterization
    \[
    \alpha(\boldsymbol{\theta}_{-n}; w) = (\alpha_1(\boldsymbol{\theta}_{-n}; w), \ldots, \alpha_n(\boldsymbol{\theta}_{-n}; w))
    \]
    for each \(w > 0\), where \(\boldsymbol{\theta}_{-n} = [\theta_1, \ldots, \theta_{n-1}]\).
    
    \begin{enumerate}
        \item Compute \(\abs{d\alpha(\boldsymbol{\theta}_{-n}; w)} = \sqrt{\det[J_{\alpha}(\boldsymbol{\theta}_{-n})^{\top} J_{\alpha}(\boldsymbol{\theta}_{-n})]} \), where \(J_{\alpha}(\boldsymbol{\theta}_{-n})\) denotes the Jacobian of \(\alpha\) evaluated at \(\boldsymbol{\theta}_{-n}\).
    
        \item Compute \(\{u_{i,w}\}_{i = 1}^{n-1}\) as
        \begin{align*}
            u_{1,w} &= u_1(\theta_1; w) \\
            &= \frac{\int_{\inf \Theta_1}^{\theta_1} \int_{\boldsymbol{\theta}_{2:n-1}} \abs{d\alpha(x, \theta_{2:n-1}; w)} \, \text{d}x \, \text{d}\theta_2 \cdots \text{d}\theta_{n-1}}{\int_{\boldsymbol{\Theta}} \abs{d\alpha(x, \theta_{2:n-1}; w)} \, \text{d}x \, \text{d}\theta_2 \cdots \text{d}\theta_{n-1}}, \\
            u_{i,w} &= u_i(\theta_i; \theta_{1:i-1}, w) \\
            &= \frac{\int_{\inf \Theta_i}^{\theta_i} \int_{\boldsymbol{\theta}_{i+1:n-1}} \abs{d\alpha(\theta_{1:i-1}, x, \theta_{i+1:n-1}; w)} \, \text{d}x \, \text{d}\theta_{i+1} \cdots \text{d}\theta_{n-1}}{\int_{\boldsymbol{\theta}_{i:n-1}} \abs{d\alpha(\theta_{1:i-1}, x, \theta_{i+1}, \ldots, \theta_{n-1}; w)} \, \text{d}x \, \text{d}\theta_{i+1} \cdots \text{d}\theta_{n-1}}
        \end{align*}
        for \(i = 2, \ldots, n-1\). Here, \(\theta_{i:j}\) denotes the vector \([\theta_i, \theta_{i+1}, \ldots, \theta_j]\) and \(\int_{\boldsymbol{\theta}_{i:j}}\) denotes an integral over the Cartesian product of \(\boldsymbol{\theta}_i, \boldsymbol{\theta}_{i+1}, \ldots, \boldsymbol{\theta}_j\). \label{item:u_i}
    
        \item Follow Equations \eqref{analytic_2d_GP} and \eqref{eq:WCP2_beta} to obtain the WCP prior density of \(\boldsymbol{\theta}\). \label{item:step3}
    \end{enumerate}    
\end{recipe}

For models with two parameters, the level set $S_{w,\boldsymbol{\theta}}$ is a level curve and $\boldsymbol{\Theta} \subset \mathbb{R}^2$.
We can create a Cartesian coordinate system for the two parameters, $\theta_1$ and $\theta_2$, with each one representing one axis.
The two parameters corresponding to a base model should be a point in that coordinate system.
Without loss of generality, we can choose a parameterization so that the point is the origin of the coordinate system.
Let $w = W_p(\theta_1,\theta_2)$ be the Wasserstein distance between a flexible model with parameters at $(\theta_1,\theta_2)$ and the base model.
Proposition~\ref{prp:recipe_multi_pcwprior} requires us to assign uniform distributions over each level curve and a (truncated) exponential distribution on the Wasserstein distance.
An area-preserving parameterization in this case means a parameterization of the level curve by arc-length.

In some cases, it might be difficult to find a parameterization by arc-length of level curves. The following recipe provides a solution for how to derive the WCP$_p$ priors in the bivariate case when each level curve is a graph of a function. 

\begin{recipe}
	\label{arclength_recipe}
	Suppose that for each $w>0$, the level curve $S_{w,\boldsymbol{\theta}}$ is compact and is a graph of a function. In particular, by exchanging the order of $\theta_1$ and $\theta_2$ if necessary, it can be parameterized as $\alpha(\theta_1,\theta_2;w) = (\theta_1,f(\theta_1;w))$, where
	$f(\cdot;w)$ is a function of $\theta_1$ that depends on the Wasserstein distance $w$. 
	Let $(\theta_1^o,\theta_1^e)$ denote the domain of $\theta_1$
	and let ${s = u_1(\theta_1;w)}$ denote the arc length from $(\theta_1^o,f(\theta_1^o;w))$ to $(\theta_1,f(\theta_1;w))$. Recall that ${c := \sup_{\boldsymbol{\theta}\in\boldsymbol{\Theta}}W_p(\boldsymbol{\theta})}$.
	The steps to derive the bivariate WCP$_p$ prior are:
	\begin{enumerate}
		\item Compute $u_1(\theta_1;w)$ and the total arc length $l(w)$ as
		\begin{align*}
		u_1(\theta_1;w) &= \int_{\theta_1^o}^{\theta_1}\sqrt{1+\left({\mathrm{d}f(x;w)}/{\mathrm{d}x}\right)^2} \mathrm{d}x,\\
		l(w) &= \int_{\theta_1^o}^{\theta_1^e} \sqrt{1+\left({\mathrm{d}f(x;w)}/{\mathrm{d}x}\right)^2} \mathrm{d}x.
        \end{align*}
		
		\item Compute the Jacobian determinant
		
		\begin{equation}\label{eq:determinant_jacobian}
		   \det J_{W_p,u_1}(\theta_1,\theta_2) = 
			\left|\begin{array}{cc}
				\frac{\partial W_p(\theta_1,\theta_2)}{\partial \theta_1} & \frac{\partial W_p(\theta_1,\theta_2)}{\partial \theta_2}\\
				\frac{\partial u_1(\theta_1;W_p(\theta_1,\theta_2))}{\partial \theta_1} & \frac{\partial u_1(\theta_1;W_p(\theta_1,\theta_2))}{\partial \theta_2} 
			\end{array}\right|.
		\end{equation}
		
		\item \label{arclength_recipe:3} Compute the density of the bivariate WCP prior of $(\theta_1,\theta_2)\in\boldsymbol{\Theta}$ as
		\begin{equation}
			\label{2dpcwdensity}
			\pi_{\theta_1,\theta_2}(\theta_1, \theta_2) = \frac{\eta\exp(-\eta W_p(\theta_1,\theta_2))}{1-\exp(-\eta c)}\frac{\abs{\det J_{W_p,u_1}(\theta_1,\theta_2)}}{{l(W_p(\theta_1,\theta_2))}}. 
		\end{equation}
	\end{enumerate}
\end{recipe}
It can be noted that a parameterization by arc-length is $\gamma(s) = \alpha \circ u_1^{-1}(s)$. This might be difficult to compute because it involves the inversion of $u_1$, but this inversion is not needed in order to compute the WCP prior through the recipe. 

We have the following proposition, whose proof is immediate, showing the expression of the WCP$_p$ prior in the sense of Definition \ref{def:multi_wcp_ref_par} for regions that are not necessarily of the form $I\times J$, where $I$ and $J$ are intervals, but a WCP$_p$ prior is available for some reference parameterization. 

\begin{proposition}\label{cor:general_regions}
	Let $\hat{\boldsymbol{\Theta}} = \hat{\boldsymbol{\Theta}}_1\cup\cdots\cup\hat{\boldsymbol{\Theta}}_k\subset\mathbb{R}^2$, where $k\in\mathbb{N}$. Further, assume that for each $j$, $\hat{\boldsymbol{\Theta}}_j$ satisfies the assumptions of Definition \ref{def:multi_wcp_ref_par}, for some function $g_j(\cdot)$. Finally, assume that for each $i\neq j$, $\lambda(\hat{\boldsymbol{\Theta}}_i\cap \hat{\boldsymbol{\Theta}}_j)=0$, where $\lambda(\cdot)$ is the Lebesgue measure and $i,j=1,\ldots,k$. Then,
	$$\pi_{{\theta}_1,{\theta}_2}(\hat{\theta}_1,\hat{\theta}_2) = \sum_{j=1}^k |\det J_{g_j}(\hat{\theta}_1,\hat{\theta}_2)| \pi(g_j(\hat{\theta}_1,\hat{\theta}_2)) \mathbbm{1}_{\hat{\theta}_1,\hat{\theta}_2\in \hat{\boldsymbol{\Theta}}_j}.$$
	In particular, one can also apply Recipe \ref{arclength_recipe} for conic regions such as the following: $$\{(x,y)\in\mathbb{R}^2: (x,y) = (r\cos(\theta),r\sin(\theta)), r >0, \theta \in [0,\phi]\}, \phi \in (0,2\pi).$$
\end{proposition}
\begin{proof}
    The result follows directly from the change of variables formula, together with an application of Sard's theorem to drop the requirement that ${\det J_\Psi(\cdot)\neq 0}$ \cite[e.g.,][p.72]{spivak1965calculus}. 
\end{proof}

\section{WCP$_2$ priors for linear regression}\label{app:wcp4regression}

We first show how to derive the step-wise WCP$_2$ for $\boldsymbol{\beta} = (\beta_1, \ldots, \beta_n)^\top \in \mathbb{R}^n$. 
The simplest model is when $\boldsymbol{\beta} = \boldsymbol{0}$, thus $\boldsymbol{\Theta}_b = \boldsymbol{0}$.
For $\beta_{1}$, we compute its WCP$_2$ prior conditioned on all the other parameters taking their base model value $\beta_i = 0, i = 2,...,n$.
That is to penalize the Wasserstein-2 distance between: $\mathcal{N}(\boldsymbol{0}, \sigma^2 \boldsymbol{I}_N)$ and $\mathcal{N}(\beta_1 \boldsymbol{X}_{(1)}, \sigma^2 \boldsymbol{I}_N)$, where $\boldsymbol{X}_{(1)}$ denotes the first column of $\boldsymbol{X}$ and $\boldsymbol{I}_N$ is an $N\times N$ identity matrix.
We then derive a WCP$_2$ prior for each $\beta_{i}$ given $\beta_{1}, \ldots, \beta_{i-1}, \beta_{i+1} = 0, \ldots, \beta_{m} = 0$ for $i = 2,...,n$.
That is to penalize the Wasserstein-2 distance between: $\mathcal{N}(\sum_{j=1}^{i - 1} \boldsymbol{X}_{(j)} \beta_{j}, \sigma^2 \boldsymbol{I}_N)$ and $\mathcal{N}(\sum_{j=1}^{i} \boldsymbol{X}_{(j)} \beta_{j}, \sigma^2 \boldsymbol{I}_N)$, where $\boldsymbol{X}_{(i)}$ denotes the column $i$ of $\boldsymbol{X}$.
Thus, the condition WCP$_2$ for $\beta_i$ is \[\pi_{\beta_i}(\beta_i) = \frac{\eta_i}{2} \| \boldsymbol{X}_{(i)} \|_{\mathbb{R}^N} \exp \left( -\eta_i \cdot \norm{ \boldsymbol{X}_{(i)} }_{\mathbb{R}^N} |\beta_i| \right)\] for $i = 1,...,n$, where $\eta_i$ is a user-specified parameter.
The final step-wise WCP$_2$ for $\boldsymbol{\beta}$ is \(\pi_{\boldsymbol{\beta}}(\boldsymbol{\beta}) = \prod_{i = 1}^n \pi_{\beta_i}(\beta_i)\).
This proves the step-wise WCP$_2$ prior expression for $\boldsymbol{\beta}$.

We will now derive the multivariate WCP$_2$ prior for $\boldsymbol{\beta}$. Begin by observing that $$W_2(\boldsymbol{\beta}) = W_2(\mathcal{N}(\boldsymbol{X} \boldsymbol{\beta}, \sigma^2 \boldsymbol{I}_N), \mathcal{N}( \boldsymbol{0}, \sigma^2 \boldsymbol{I}_N)) = \norm{\boldsymbol{X}\boldsymbol{\beta}}_{\mathbb{R}^n} = \sqrt{\boldsymbol{\beta}^\top \boldsymbol{X}^\top \boldsymbol{X} \boldsymbol{\beta}}.$$ Since $\boldsymbol{X}^\top \boldsymbol{X}$ is a symmetric and non-negative definite matrix, the level sets $S_{w,\boldsymbol{\beta}}$ are given by the ellipsoids of the form $\{\boldsymbol{\beta}\in\mathbb{R}^n: \boldsymbol{\beta}^\top \boldsymbol{X}^\top \boldsymbol{X} \boldsymbol{\beta} = w^2\}$. Let $\boldsymbol{X}^\top \boldsymbol{X} = \boldsymbol{P}^\top \boldsymbol{\Lambda} \boldsymbol{P}$ be the eigendecomposition of $\boldsymbol{X}^\top \boldsymbol{X}$, where $\boldsymbol{P}$ is an orthogonal matrix and $\boldsymbol{\Lambda}$ is a diagonal matrix with non-negative eigenvalues. Then, the level sets can be parameterized as ${\{\boldsymbol{\beta}\in\mathbb{R}^n: \boldsymbol{\beta}^\top \boldsymbol{P}^\top \boldsymbol{\Lambda} \boldsymbol{P} \boldsymbol{\beta} = w^2\}}$. By defining $\widetilde{\boldsymbol{\beta}} = \boldsymbol{P} \boldsymbol{\beta}$, with $\widetilde{\boldsymbol{\beta}} = (\widetilde{\beta}_1, \ldots, \widetilde{\beta}_n)$, we can rewrite the level sets as 
\begin{equation}
    \label{eq:ellipsoid_level_set}
    S_{w,\boldsymbol{\beta}} = \left\{\widetilde{\boldsymbol{\beta}}\in\mathbb{R}^n: \sum_{j=1}^n\lambda_j\widetilde{\beta}_j^2 = w^2\right\}.
\end{equation}
Now, we can use the spherical parameterization (also known as hyperspherical parameterization or polar parameterization) of the ellipsoids:
\begin{align*}
\widetilde{\beta}_1 &= \frac{w\cos(\varphi_1)}{\sqrt{\lambda_1}}, \quad \widetilde{\beta}_2 = \frac{w\sin(\varphi_1)\cos(\varphi_2)}{\sqrt{\lambda_2}} \\
&\vdots \\
\widetilde{\beta}_{n-1} &= \frac{w\sin(\varphi_1)\sin(\varphi_2)\cdots\sin(\varphi_{n-2})\cos(\varphi_{n-1})}{\sqrt{\lambda_{n-1}}} \\
\widetilde{\beta}_n &= \frac{w\sin(\varphi_1)\sin(\varphi_2)\cdots\sin(\varphi_{n-2})\sin(\varphi_{n-1})}{\sqrt{\lambda_n}}
\end{align*}
where $\varphi_1, \ldots, \varphi_{n-2} \in [0, \pi]$ and $\varphi_{n-1} \in [0, 2\pi)$. Furthermore, observe that $\boldsymbol{\beta} = \boldsymbol{P}^\top \widetilde{\boldsymbol{\beta}}$, so that by letting $\boldsymbol{\varphi} = (\varphi_1, \ldots, \varphi_{n-1})$, we have
\begin{equation}\label{eq:Xw_ellipsoid}
X_w(\boldsymbol{\varphi}) = w \boldsymbol{P}^\top \begin{pmatrix}
\cos(\varphi_1)/\sqrt{\lambda_1} \\
\sin(\varphi_1)\cos(\varphi_2)/\sqrt{\lambda_2} \\
\vdots \\
\sin(\varphi_1)\sin(\varphi_2)\cdots\sin(\varphi_{n-2})\cos(\varphi_{n-1})/\sqrt{\lambda_{n-1}} \\
\sin(\varphi_1)\sin(\varphi_2)\cdots\sin(\varphi_{n-2})\sin(\varphi_{n-1})/\sqrt{\lambda_n}.
\end{pmatrix}.
\end{equation}
In particular, we have that 
$$X_w(\boldsymbol{\varphi}) = w \boldsymbol{P}^\top \boldsymbol{\Lambda}^{-1/2} Z_w(\boldsymbol{\varphi}),$$
where $Z_w(\boldsymbol{\varphi})$ is the usual spherical parameterization of the unit sphere:
$$Z_w(\boldsymbol{\varphi}) =  \begin{pmatrix}
    \cos(\varphi_1) \\
    \sin(\varphi_1)\cos(\varphi_2) \\
    \vdots \\
    \sin(\varphi_1)\sin(\varphi_2)\cdots\sin(\varphi_{n-2})\cos(\varphi_{n-1}) \\
    \sin(\varphi_1)\sin(\varphi_2)\cdots\sin(\varphi_{n-2})\sin(\varphi_{n-1}).
    \end{pmatrix}.$$
It is well-known that the map $(w,\boldsymbol{\varphi}) \mapsto Z_w(\boldsymbol{\varphi})$ is a local diffeomorphism, thus the map defined by ${\boldsymbol{\Psi}: (w,\boldsymbol{\varphi}) \mapsto X_w(\boldsymbol{\varphi})}$ is a local diffeomorphism. Therefore, from Proposition~\ref{prp:single_param_wcp_D_matrix}, the WCP$_2$ prior for $\boldsymbol{\beta}$ is given by
\begin{equation}
    \pi_{\boldsymbol{\beta}}(\boldsymbol{\beta}) = \frac{\eta\exp\left(-\eta \|\boldsymbol{X}\boldsymbol{\beta}\|_{\mathbb{R}^n}\right)}{\hbox{Area}_{n-1}(S_{W_2(\boldsymbol{\beta}), \boldsymbol{\beta}})\sqrt{D(W_2(\boldsymbol{\beta}), X_{W_2(\boldsymbol{\beta})}^{-1}(\boldsymbol{\beta}))}}.\nonumber
\end{equation}
Now, observe that $\partial X_w(\boldsymbol{\varphi})/\partial w = X_w(\boldsymbol{\varphi})/w$, so that
$$
    D(w, \boldsymbol{\varphi}) = \frac{1}{w^2} X_{w}(\boldsymbol{\varphi})^\top \boldsymbol{P}_{w}(\boldsymbol{\varphi}) X_{w}(\boldsymbol{\varphi}),
$$
where 
$$
    \boldsymbol{P}_{w}(\boldsymbol{\varphi}) = \boldsymbol{I}_{n-1} - J_{X_{w}}(\boldsymbol{\varphi})(J_{X_{w}}(\boldsymbol{\varphi})^\top J_{X_{w}}(\boldsymbol{\varphi}))^{-1}\!J_{X_{w}}(\boldsymbol{\varphi})^\top.
$$
Thus, to compute $D(W_2(\boldsymbol{\beta}), X_{W_2(\boldsymbol{\beta})}^{-1}(\boldsymbol{\beta}))$ we need the following ingredients: $W_2(\boldsymbol{\beta}) = \|\boldsymbol{X}\boldsymbol{\beta}\|_{\mathbb{R}^n}$, $X_{W_2(\boldsymbol{\beta})}(X_{W_2(\boldsymbol{\beta})}^{-1}(\boldsymbol{\beta})) = \boldsymbol{\beta}$, the angles $\boldsymbol{\varphi} = (\varphi_1, \ldots, \varphi_{n-1})$ that are computed as:
   \[
   \varphi_k = \arctan2\left(\sqrt{\sum_{i=k+1}^n \lambda_i \widetilde{\beta}_i^2}, \sqrt{\lambda_k}\widetilde{\beta}_k\right) \quad \text{for } k = 1, \ldots, n-1.
   \]
where $\arctan2(\cdot, \cdot)$ is the two-argument arctangent function, and $J_{X_w}(\boldsymbol{\varphi}) = w \boldsymbol{P}^\top \boldsymbol{\Lambda}^{-1/2} J_{Z_w}(\boldsymbol{\varphi})$, where $J_{Z_w}(\boldsymbol{\varphi}) = (J_{jk})_{j=1,k=1}^{j=n,k=n-1}$ is the Jacobian matrix of the spherical parameterization, which is given by:
$$J_{jk} = 
\begin{cases} 
0 & \text{if } j < k, \\
-w \sin(\varphi_k) \prod_{m=1}^{k-1} \sin(\varphi_m) & \text{if } j = k, \\
w \cos(\varphi_k) \cos(\varphi_j) \prod_{\substack{m=1 \\ m \neq k}}^{j-1} \sin(\varphi_m) & \text{if } k+1 \leq j \leq n-1, \\
w \cos(\varphi_k) \prod_{\substack{m=1 \\ m \neq k}}^{n-1} \sin(\varphi_m) & \text{if } j = n.
\end{cases}$$ 
With these ingredients, we can now compute $D(W_2(\boldsymbol{\beta}), X_{W_2(\boldsymbol{\beta})}^{-1}(\boldsymbol{\beta}))$. 
All that remains is to compute the area of the level set $S_{W_2(\boldsymbol{\beta}), \boldsymbol{\beta}}$. To this end, we have:
\begin{lemma} 
    \label{lem:area_ellipsoid}
    Let $S_{w,\boldsymbol{\beta}}$ be the ellipsoid in \eqref{eq:ellipsoid_level_set}. Then, its the surface area is 
    $$
    \hbox{Area}_{n-1}(S_{W_2(\boldsymbol{\beta}), \boldsymbol{\beta}}) = \frac{2\pi^{n/2} W_2(\boldsymbol{\beta})^{n-1}}{\Gamma\left(\frac{n}{2}\right)\prod_{k=1}^{n-1} \lambda_{k}} F_D^{(n-1)}\left(-\frac{1}{2};\frac{1}{2};\frac{n}{2};\alpha_1,\ldots,\alpha_{n-1}\right),
    $$
    where $\alpha_i = 1 - \frac{\lambda_{(i)}^2}{\lambda_{(n)}^2}$, $\lambda_{(1)}, \ldots, \lambda_{(n-1)}$ are the ordered $\lambda_i$ in increasing order and $F_D^{(n-1)}$ is the Lauricella hypergeometric function is defined as follows:
    \begin{align*}
    F_D^{(n)}(a; b; c; x_1, \ldots, x_n) = 
    \frac{\Gamma(c)}{\Gamma(a)\Gamma(c-a)}\int_0^1 \frac{(1-u)^{c-a-1}}{u^{1-a}}\prod_{i=1}^n (1-ux_i)^{-b}du,
    \end{align*}
    where $a, b, c, x_i \in \mathbb{R}$, with $|x_i| < 1$, $i = 1, \ldots, n$, $a > 0$, $c > a$, and $n\in\mathbb{N}$.    
\end{lemma}
\begin{proof}
The result follows by using the explicit expression for the surface area of ellipsoids given in \cite{rivin2007surface} and the identity given in \citet[Proposition 2.4]{krason2020linear}\footnote{This expression has been obtained in another fashion in \url{https://analyticphysics.com/Higher\%20Dimensions/Ellipsoids\%20in\%20Higher\%20Dimensions.htm}}.
\end{proof}


\section{Proofs of results}\label{app:proofs}
In this section we provide proofs of the results. If a proof is omitted is because it is straightforward.

\begin{proof}[Proof of Proposition~\ref{example1}]
    It is a direct consequence of Proposition \ref{prp:recipe_multi_pcwprior} given in Appendix \ref{app:wasserstein}.
\end{proof}

\begin{proof}[Proof of Proposition~\ref{example2}]
    Since \(\mu_b\) and \(\mu_m\) differ only by a location parameter, and using Remark~\ref{remarkwformula} in Appendix \ref{app:wasserstein}, we have \(W_p(\mu_b, \mu_m) = |m|\). Because \(W_p(\mu_b, \mu_m)\) tends to infinity as \(m \to \infty\), Definition~\ref{pcwprior_defi} directly provides the desired expression.
\end{proof}
\begin{proof}[Proof of Proposition~\ref{prp:AR_wcp_density}]
Let $\boldsymbol{\Sigma}_0$ and $\boldsymbol{\Sigma}$ denote the covariance matrices of the base and flexible models for the process $\{X_t, t=1, \dots, n\}$, with $n < \infty$. These matrices are given, respectively, by $\boldsymbol{\Sigma}_0$ with all entries equal to $\sigma^2$ and $\boldsymbol{\Sigma} = (\phi^{|i-j|} \sigma^2)_{i,j=1}^n$.

To compute the squared Wasserstein-2 distance $W_2^2(\mu, \mu_b)$ between the corresponding Gaussian measures, we use \eqref{w2dist2gaussianfd}:
$$
W_2^2(\mu, \mu_b) = 2n\sigma^2 - 2\, \text{tr}\left\{ \left( \boldsymbol{\Sigma}_0^{1/2} \boldsymbol{\Sigma} \boldsymbol{\Sigma}_0^{1/2} \right)^{1/2} \right\}.
$$
Since $\boldsymbol{\Sigma}_0$ is a rank-1 matrix with all entries equal to $\sigma^2$, its square root is $\boldsymbol{\Sigma}_0^{1/2} = \frac{\sigma}{\sqrt{n}} \mathbf{e}\mathbf{e}^T$, where $\mathbf{e}$ is the $n$-dimensional vector of ones. Then, the product $\boldsymbol{\Sigma}_0^{1/2} \boldsymbol{\Sigma} \boldsymbol{\Sigma}_0^{1/2}$ simplifies to 
$
\frac{\sigma^2}{n} \left( \mathbf{e}^T \boldsymbol{\Sigma} \mathbf{e} \right) \mathbf{e}\mathbf{e}^T,
$
which is a scalar multiple of $\boldsymbol{\Sigma}_0$. 

The quantity $\mathbf{e}^T \boldsymbol{\Sigma} \mathbf{e}$ is the sum of all entries in $\boldsymbol{\Sigma}$. Since $\boldsymbol{\Sigma}$ is Toeplitz with entries $\phi^{|i-j|}$, we have
\[
\mathbf{e}^T \boldsymbol{\Sigma} \mathbf{e} = \sigma^2 \left( n + 2 \sum_{k=1}^{n-1} (n - k)\phi^k \right)
= \sigma^2 \frac{n(1-\phi^2) - 2\phi(1-\phi^n)}{(1-\phi)^2},
\]
where we in the second equality used the formula for the sum of a finite weighted geometric series.
Taking the square root of the scalar multiple and computing its trace, we obtain
\[
\text{tr}\left\{ \left( \boldsymbol{\Sigma}_0^{1/2} \boldsymbol{\Sigma} \boldsymbol{\Sigma}_0^{1/2} \right)^{1/2} \right\} = \frac{\sigma \sqrt{n(1-\phi^2) - 2\phi(1-\phi^n)}}{1-\phi}.
\]

Substituting this into the expression for $W_2^2(\mu, \mu_b)$ gives the final result:
$$
W_2^2(\mu, \mu_b) = 2\sigma^2 \left( n - \frac{\sqrt{n(1-\phi^2) - 2\phi(1-\phi^n)}}{1-\phi} \right).
$$

Finally, note that $W_2(\mu, \mu_b)$ increases as $\phi$ decreases, and remains bounded above by the constant ${c = \sigma \left( 2n - \sqrt{2} \sqrt{1 - (-1)^n} \right)^{1/2} < \infty}$. The result then follows directly by applying Remark~\ref{rem:wcp_oneside} and simplifying.

\end{proof}

\begin{proof}[Proof of Proposition~\ref{prp:gen_pareto_wcp_density}]
    By Remark \ref{remarkwformula}, the Wasserstein-1 distance between the base model and the flexible model with parameter $\xi$ is $W_1(\xi) = \xi/(1-\xi)$. The result follows by using Remark~\ref{rem:wcp_oneside}.
\end{proof}

\begin{proof}[Proof of Proposition~\ref{2d_gaussian_wcp}]
	By Remark \ref{remarkwformula}, the Wasserstein-2 distance between the base measure (Dirac measure concentrated at zero) and a flexible model with mean $m$ and standard deviation $\sigma$ is $W_2(m,\sigma) = (m^2 + \sigma^2)^{1/2}$, which coincides with the Euclidean distance on $\mathbb{R}\times (0,\infty)$. 
	For any fixed value of $W_2(m,\sigma) = w > 0$, the level curve $S_{w,\boldsymbol{\theta}}$ is a semi-circle with radius $w$. A parameterization for $S_{w,\boldsymbol{\theta}}$ is given by $X_{w,\boldsymbol{\theta}}(u) = (u,\sqrt{w^2 - u^2})$. Now, observe that $X_w^{-1}(m,\sigma) = m$, which gives us $\boldsymbol{\Phi}(m,\sigma) = (\sqrt{m^2+\sigma^2}, m)$, and by Remark~\ref{remark:diffeomorphism}, it is enough to show that $\boldsymbol{\Phi}$ is a local diffeomorphism. First, observe that the domain of $\boldsymbol{\Phi}$ is $\mathbb{R}\times (0,\infty)$, which is an open set. Second, the Jacobian matrix of $\boldsymbol{\Phi}$ is 
    $$J_{\boldsymbol{\Phi}}(m,\sigma) = \begin{pmatrix}
    \frac{m}{\sqrt{m^2+\sigma^2}} & \frac{\sigma}{\sqrt{m^2+\sigma^2}} \\
    1 & 0
    \end{pmatrix} \Rightarrow |\det J_{\boldsymbol{\Phi}}(m,\sigma)| = \frac{\sigma}{\sqrt{m^2+\sigma^2}}.
    $$
    Further, we have that $J_{X_w}(u) = (1, u/\sqrt{w^2-u^2})$ so that 
    $$J_{X_{W_p(m,\sigma)}}(X_{W_p(m,\sigma)}^{-1}(m,\sigma))= J_{\sqrt{m^2+\sigma^2}}(m)= \left(1, \frac{m}{\sigma}\right)$$
    and
    $$\boldsymbol{G}(m,\sigma) = \sqrt{1 + \left(\frac{m}{\sigma}\right)^2} = \frac{\sqrt{m^2+\sigma^2}}{\sigma}.$$
    Finally, $\hbox{Area}_1(S_{w,\boldsymbol{\theta}})$ is the arc-length of a semi-circle with radius $W_p(m,\sigma)$, so that we have $\hbox{Area}_1(S_{w,\boldsymbol{\theta}}) = \pi\sqrt{m^2+\sigma^2}$. Combining all the elements, we obtain
    $$\pi(m,\sigma) = \frac{\eta\exp(-\eta (m^2+\sigma^2)^{1/2})}{{\pi}(m^2+\sigma^2)^{1/2}} .$$
\end{proof}

\begin{remark}
    Alternatively, one could use the well-known parameterization of $S_{w,\boldsymbol{\theta}}$ by arc-length given by $X_w(\varphi) = (w\cos(\varphi), w\sin(\varphi))$, where $\varphi \in (0,\pi)$ denotes the polar angle. In this case, we have $\boldsymbol{G}(\boldsymbol{\theta})=1$.
\end{remark}

\begin{proof}[Proof of Proposition~\ref{2d_GP_WCP}]
        In this case, the flexible models correspond to the parameters $\xi\in[0,1)$ and $\sigma\in (0,+\infty)$.
        Let $Z$ follow a generalized Pareto distribution with $\sigma=1$ and $\xi$. Then, generalized Pareto density $f_{\xi,\sigma}(y)$ is a location-scale family generated by $X$, where the scale parameter is $\sigma$ and the location parameter is $0$.
        According to Proposition \ref{loc-scale-prop}, the Wasserstein-1 distance between the base measure, which is Dirac, and a flexible model is ${W_1(\xi,\sigma) = \sigma \mathbb{E}X = \frac{\sigma}{1-\xi}}$. 
        By fixing $W_1(\xi,\sigma)$ to a positive value $w$, we obtain a level curve that can be parameterized by $\alpha(\sigma) = (\sigma, 1-\frac{\sigma}{w})$, which is a straight line in the Cartesian coordinate system.
        By following Recipe~\ref{arclength_recipe}, let $s = u_1(\sigma;w)$ denote the partial arc length of the level curve from the point $(0,1)$ to $(\sigma,1-\frac{\sigma}{w})$ as a function of $\sigma$. 
        We have that $u_1(\sigma;w) = \frac{\sigma}{w}{\sqrt{w^2+1}}$.
        Therefore, the full arc length of the level curve is $l(w) = \sqrt{w^2+1}$.
        By Recipe \ref{arclength_recipe}:\eqref{arclength_recipe:3}, we obtain the $\text{WCP}_1$ prior of $(\xi,\sigma)$ as
        \begin{align*}
            \pi_{\sigma,\xi}(\sigma,\xi) &= \frac{\sqrt{(1-\xi)^2+ \sigma^2}}{(1-\xi)^2}\pi_{w,s}(\frac{\sigma}{1-\xi},\sqrt{(1-\xi)^2 + \sigma^2})\\
            &= \frac{\eta}{1-\xi}\exp(-\eta\frac{\sigma}{1-\xi}).
        \end{align*}
    \end{proof} 

\bibliographystyle{abbrvnat}
\bibliography{reference}
\end{document}